\documentclass[lettersize,journal]{IEEEtran}
\usepackage{amsmath,amsfonts,mathtools}
\usepackage{algorithmic}
\usepackage{algorithm}
\usepackage{array}
\usepackage[caption=false,font=normalsize,labelfont=sf,textfont=sf]{subfig}
\usepackage{textcomp}
\usepackage{stfloats}
\usepackage{url}
\usepackage{verbatim}
\usepackage{graphicx}
\usepackage{cite}
\usepackage{color,bm,upgreek,enumitem,amssymb}

\newcommand{\mf}[1]{\mathbf{#1}}
\DeclareMathOperator{\Ad}{Ad}
\DeclareMathOperator{\ad}{ad}
\DeclareMathOperator{\tr}{tr}
\DeclareMathOperator{\var}{var}

\usepackage{amsthm}
\newtheorem{theorem}{Theorem}
\newtheorem{lemma}{Lemma}
\newtheorem{corollary}{Corollary}
\newtheorem{remark}{Remark}

\begin{document}

\title{Parameter Estimation on Homogeneous Spaces}

\author{Shiraz Khan, Gregory S. Chirikjian,~\IEEEmembership{Fellow,~IEEE}%
\thanks{\textsuperscript{\textdagger}The authors are with the Department of Mechanical Engineering, University of Delaware, Newark, DE 19716, USA (e-mail: \tt{shiraz@udel.edu}, \tt{gchirik@udel.edu}).}%
\textsuperscript{\textdagger}
}


\maketitle


\begin{abstract}
The \textit{Fisher Information Metric (FIM)} and the associated \textit{Cram\'er-Rao Bound (CRB)} are fundamental tools in statistical signal processing, which inform the efficient design of experiments and algorithms for estimating the underlying parameters. In this article, we investigate these concepts for the case where the parameters lie on a \textit{homogeneous space}. Unlike the existing Fisher-Rao theory for general Riemannian manifolds, our focus is to leverage the group-theoretic structure of homogeneous spaces, which is often much easier to work with than their Riemannian structure. The FIM is characterized by identifying the homogeneous space with a coset space, the group-theoretic CRB and its corollaries are presented, and its relationship to the Riemannian CRB is clarified. The application of our theory is illustrated using two examples from engineering: {(i) estimation of the pose of a robot} and {(ii) sensor network localization}. In particular, these examples demonstrate that homogeneous spaces provide a natural framework for studying statistical models that are \textit{invariant} with respect to a group of symmetries.
\end{abstract}

\begin{IEEEkeywords}
Cram\'er-Rao bound, Lie groups, state estimation
\end{IEEEkeywords}

\section{Introduction}
\IEEEPARstart{M}{any} engineering problems require the estimation of some underlying parameter(s) $\theta \in \Theta$ based on noisy data, where $\Theta$ is called the \textit{parameter space}.
In the medical sciences, one infers the parameters of a physiological model using medical tests and imaging techniques, while accounting for random variability in the observed data. In supervised deep learning, the training dataset represents the observation, such that the parameters of a neural network are estimated to best fit this observation.
It is well known that the parameter spaces of such problems can be equipped with a geometric structure which is canonical (i.e., not based on arbitrary choices), by way of the \textit{Fisher information metric (FIM)}.
The FIM enables one to interpolate and measure distances between probability distributions \cite{costa2015fisher}, quantify the uncertainty of estimators using the \textit{Cram\'er-Rao bound (CRB)}\cite{smith2005covariance}, and regularize variational inference algorithms \cite{li2020fisher}. An empirical approximation of the FIM is used to accelerate gradient descent in the deep learning context \cite{martens2020new,zhang2019fast,karakida2020understanding,faria2023fisher} on account of the invariance of the FIM to reparametrization of the statistical model, a feature that is not shared by alternative notions of statistical distance.
The CRB is of independent interest to the signal processing community, with applications in sensor selection \cite{kokke2023sensor} as well as for the design of optimal experiments \cite{zhao2018optimal} and estimators \cite{chen2021cramer,song2023intelligent}.

Classical analyses of the FIM and CRB have assumed the parameter space $\Theta$ to be a \textit{vector space}. However, many parametric models of interest have additional constraints that force the parameters to lie on a lower dimensional submanifold of a vector space. 
For e.g., the covariance of an $n$-dimensional Gaussian distribution lies in the space of $n\times n$ symmetric positive definite (SPD) matrices \cite{smith2005covariance},
the end effector of a robotic arm takes values on the Lie group $SE(3)$ \cite{wang2006error}, the heading/bearing of an aircraft is described by a point on the two-dimensional sphere \cite{dogandzic2001cramer}, and so on. In each of these examples, 
the unconstrained problem is typically ill-posed and does not admit a unique solution.
 As an alternative to constrained parameter estimation on vector spaces (which is an extrinsic viewpoint of the problem), one can consider an unconstrained estimation problem on \textit{manifold}-valued parameters (i.e., the intrinsic viewpoint). The pioneering work of \cite{smith2005covariance} derived an approximate Riemannian CRB for manifold-valued parameters, while \cite{boumal2013intrinsic} extended this theory to \textit{quotient manifolds} that can be realized as Riemannian submersions. Group-theoretic variants of the FIM and CRB were developed in \cite{chirikjian2010information,bonnabel2015intrinsic,bonnabel2015intrinsic2,boumal2014cramer,solo2020cramer} and thereafter applied to estimator design problems in robotics \cite{chen2021cramer}. While the aforementioned works ignored the higher-order curvature terms in their calculation, it was recently shown that an exact group-theoretic CRB can be derived for parameter estimation on matrix Lie groups \cite{labsir2023barankin}, which coincides with the Riemannian CRB when the group admits a bi-invariant metric\footnote{If a Lie group admits a bi-invariant metric, then its (group-theoretic) exponential map coincides with the Riemannian exponential map \cite[Corr. 3.19]{cheeger1975comparison}; the former is used to derive the group-theoretic CRB whereas the latter is used to derive the Riemannian CRB.}. 

While the specialization of the Fisher-Rao analysis from manifolds to Lie groups yields simple and tractable expressions for the FIM and CRB, Lie groups exclude many parameter spaces of interest. Spaces that typically arise in applications, such as spheres, SPD matrices, and the Grassmannian manifolds \cite{absil2008optimization}, are each examples of a larger class of spaces, namely, \textit{homogeneous spaces}. Previous works have derived the CRB for specific instances of homogeneous spaces, e.g., a CRB for SPD matrices was derived in \cite{breloy2018intrinsic}, while the CRB for the sphere has been of academic interest for decades due to its application in Direction-of-Arrival (DoA) estimation using radar and sonar arrays \cite{dogandzic2001cramer}. Despite the practical import and ubiquity of homogeneous spaces in engineering and control, the FIM and CRB for general homogeneous spaces have not been studied previously in the literature. Since homogeneous spaces are more general (but have less structure) than Lie groups, and are less general (but have more structure) than manifolds, they combine the computational ease of group-theoretic methods with the wide range of applicability of Riemannian methods.

In this article, we investigate a class of parameter estimation problems in which the parameter of interest takes values in a homogeneous space.
To address the aforementioned gaps in the literature, we make the following contributions:
\begin{enumerate}
    \item We introduce group-theoretic definitions of the FIM and estimation error for parameter estimation problems on homogeneous spaces, which are typically easier to compute than their Riemannian counterparts.
    \item By investigating the properties of the FIM, we arrive at exact and approximate group-theoretic CRBs that bound the performance of biased and unbiased estimators.
    \item The conditions under which the group-theoretic CRB coincides with the Riemannian CRB are presented. In particular, we show that there is an intrinsic definition of the \textit{variance} of the estimator for a special class of homogeneous spaces.
    \item The theory is used to design an iterative algorithm for computing the Maximum Likelihood Estimator (MLE) on homogeneous spaces, which we call the \textit{generalized Fisher scoring} algorithm.\footnote{To the authors' knowledge, the Fisher scoring algorithm for Lie groups has not been presented in the existing literature either, whereas the Euclidean variant of the algorithm is well-known \cite{Li2019}.}
    \item Using illustrative examples, we show that the theory developed in this article can be used to systematically study inference problems
    wherein the observations are invariant with respect to certain transformations of the parameters.
    This distinctive feature of the Fisher-Rao analysis on homogeneous spaces is not shared by Euclidean spaces and Lie groups.
\end{enumerate}
The homogeneous spaces of interest to us include those that arise in engineering and control applications, which are finite dimensional and can be realized as a quotient of a real matrix Lie group by a closed subgroup -- i.e., they are of the form $G/H$. Since any Lie group $G$ is topologically isomorphic to $G/\lbrace e \rbrace$ (where $e \in G$ is the identity element), the results in this article are a generalization of the existing works on the FIM and CRB on Lie groups \cite{bonnabel2015intrinsic,bonnabel2015intrinsic2,solo2020cramer,labsir2023barankin}.
Unlike the above works, wherein the FIM is realized as a matrix with respect to a fixed (either left or right-invariant) frame of vector fields, the analysis on homogeneous spaces requires us to make the distinction between the left and right-invariant frames of $G$. 

\subsubsection*{Organization of the Article}
In Section \ref{sec:preliminaries}, we formulate the parameter estimation problem, provide some classical examples, and present preliminary details. The Fisher Information Metric (FIM) is introduced in Section \ref{sec:FIM}, and some properties of the corresponding matrices are established. Section \ref{sec:CRB} develops the Cram\'er-Rao bounds (CRBs) on Lie groups and homogeneous spaces and presents various corollaries.
Section \ref{sec:applications} discusses the applications of the Fisher-Rao theory, presents the generalized Fisher scoring algorithm, and provides illustrative examples.
Finally, Section \ref{sec:conclusion} presents the conclusion.
\subsubsection*{Notation}
We use $e$ to denote the identity element of a Lie group $G$. The tangent space of $G$ at $e$ is written as $T_eG$ and identified (as a vector space) with the Lie algebra $\mathfrak g$.
Capital letters like $X,Y$, and $Z$ denote tangent vectors in $\mathfrak g$, whereas small boldface letters like $\mf x$, $\mf y$, and $\mf z$ are reserved for vectors in $\mathbb R^n$. 
Capital boldface letters represent matrices. The Einstein summation convention is used in the proofs; see \cite[p. 18]{lee2012smooth} for a review.

\section{Preliminaries}\label{sec:preliminaries}

\subsection{Problem Formulation}\label{sec:formulation}
Let $\theta \in \Theta$ be a parameter of interest and $x \in \mathcal X$ denote a measurement of $\theta$, where $\Theta$ is called the \textit{parameter space} and $\mathcal X$ is called the \textit{sample space}. We assume that $\theta$ and $x$ are related to each other via a statistical model: the probability density function (pdf) of $x$ given the parameter $\theta$ is $\bar f(x\,|\,\theta)$. The goal of the parameter estimation problem is to design an estimator $\hat \theta :\mathcal X \rightarrow \Theta$, such that $\hat\theta(x)$ is an estimate of $\theta$ based on the observation of $x$. 

We assume that the \textit{likelihood function} $\bar f(x\,|\,\cdot\,)$ is smooth (on an open set of $\Theta$ containing $\theta$) for all $x\in \mathcal X$. Furthermore, we will use the shorthand notation $\bar f(\theta)\coloneq \bar f(\,\cdot\,|\,\theta)$ throughout, treating $\bar f(\theta)$ as a function on $\mathcal X$. Not much is assumed about the sample space $\mathcal X$ except that it admits an integration measure $\mu$ that does not depend on $\theta$, such that the density $\bar f(\theta)$ is defined with respect to (w.r.t.) $\mu$.\footnote{These are the so-called \textit{regularity conditions} that are commonly imposed in the literature, allowing us to interchange the differential operators of $\Theta$ with the integration on $\mathcal X$ \cite[Sec. 1.8]{Li2019}.}

\begin{remark}[Discrete Sample Spaces]
    One can even consider the situation where the measurements are not continuous but discrete/quantized, in which case $\mu \coloneq \mu_c$ is the counting measure, $\mathcal X$ is a discrete space, and integration is replaced by summation \cite{fu2018quantized,seveso2020quantum}. 
\end{remark}

We are interested in the case where $\Theta$ is a homogeneous space equipped with a transitive left $G$-action, where $G$ is a real matrix Lie group. 
Consequently, we can identify $\Theta$ with the topological quotient of $G$ by a closed subgroup $H \subset G$, written as $\Theta \cong G/H$ \cite[Thm. 21.18]{lee2012smooth}. The space $G/H$ is an example of a \textit{coset space}. A point on $G/H$ is called a \textit{coset} (or an \textit{equivalence class}), defined as follows: 
\begin{align}
gH \coloneq \lbrace gh \,|\, h \in H \rbrace \in G/H,
\end{align}
where $g$ is called a \textit{representative} of $gH$.
If a given pair of elements $g_1, g_2 \in G$ differ by an element of $H$, then they are representatives of the same point on $G/H$:
$$g_1H = g_2H \ \Leftrightarrow\  g_1^{-1}g_2 \in H.$$In particular, $gH = ghH$ for all $h\in H$. We let $L_g$ (resp., $R_g$) denote the left (resp., right) multiplication operation of $G$, so that $L_g(h) = gh$ for $g,h \in G$.

\subsection{Some Concrete Examples}
\subsubsection*{Example 1 (Direction-of-Arrival)} Consider the problem where $\theta$ describes the Direction-of-Arrival (DoA) of a signal (e.g., radar signals being reflected from an aircraft) and $x$ is a vector of measurements of the signal obtained across a radar array situated on the ground. This model is considered in \cite{dogandzic2001cramer}, where $\theta$ is called the DoA parameter and $x$ is called the array response. In the intrinsic setting, we can represent $\theta$ as a point on the {sphere}, $\Theta \coloneq S^2$, which is topologically the same as $SO(3)/SO(2)$. 
Hence, the intrinsic description of $\theta$ is $2$-dimensional, and each point on $\Theta$ corresponds to a unique value of DoA. 


\subsubsection*{Example 2 (Symmetric Positive Definite Matrices)} Let the space of $n\times n$ {symmetric positive definite (SPD) matrices} be denoted by $\mathbb S(n)^{++}$.
In \textit{diffusion tensor imaging}, MRI observations are used to 
estimate the diffusion tensor (represented at each point as an SPD matrix) to reconstruct the internal structure of tissues, using a statistical model such as \cite[eq. 2]{magdoom2021new}. SPD matrices are also used to describe the \textit{shape} or \textit{scatter} of distributions on Euclidean spaces (i.e., $\mathcal X = \mathbb R^n$) \cite{breloy2018intrinsic}, generalizing the use of covariance matrices to describe Gaussian distributions.

We can describe $\mathbb S(n)^{++}$ as the coset space $GL(n,\mathbb R)^+/SO(n)$,\footnote{Here, $GL(n,\mathbb R)^+$ denotes the \textit{general linear group} consisting of $n\times n$ real invertible matrices whose determinant is positive \cite{chirikjian2011stochastic}.} which follows from the observation that any invertible matrix ${\bf A}\in GL(n, \mathbb R)^+$ can be uniquely decomposed using the {polar decomposition} to obtain ${\bf A} = {\bf S} {\bf R}$, where ${\bf S}\in \mathbb S(n)^{++}$ and ${\bf R}\in SO(n)$. 
We can verify the property $gH = gh H$ of coset spaces as follows. If $g = \mf A$ and $h=\tilde{\mf R}\in SO(3)$, then
$gh = {\bf A}\tilde{\bf R} = {\bf S} {\bf R}\tilde{\bf R}$. Since polar decompositions of $g$ and $gh$ both yield the same symmetric matrix, $gH$ and $ghH$ represent the same point (namely, the point $\mf S$) on $\mathbb S(n)^{++}$. A detailed account of the space $\mathbb S(n)^{++}$ treated as a homogeneous space can be found in \cite[Sec. 23.9]{gallier2020differential}.


\subsection{Differential Structure of $G/H$}\label{sec:GH-structure}
\subsubsection{The Projection Map}
Let $\pi:G \rightarrow G/H$ denote the `projection' operation $g \mapsto gH$ which sends a group element to its coset.
Our approach is to pull back the likelihood function $\bar f$ via $\pi$ to define a likelihood function $f$ on $G$:
\begin{align}
    f(x|g) \coloneq \bar f\left(x | \pi (g)\right) = \bar f (x | gH),
    \label{eq:pullback}
\end{align}
i.e., $f = \bar f \circ \pi$, where $\circ$ denotes the composition of functions (see Fig. \ref{fig:homogeneous_space}). The inverse image under $\pi$ of a point $gH \in G/H$ is written as $\pi^{-1}(gH)$. It is diffeomorphic to (i.e., has the same smooth manifold structure as) $H$, and is referred to as the \textit{fiber} on $gH$.

\begin{figure}[htbp]
    \centering
    \includegraphics[width=0.5\textwidth,trim={17.cm 6cm 17cm 2cm},clip]{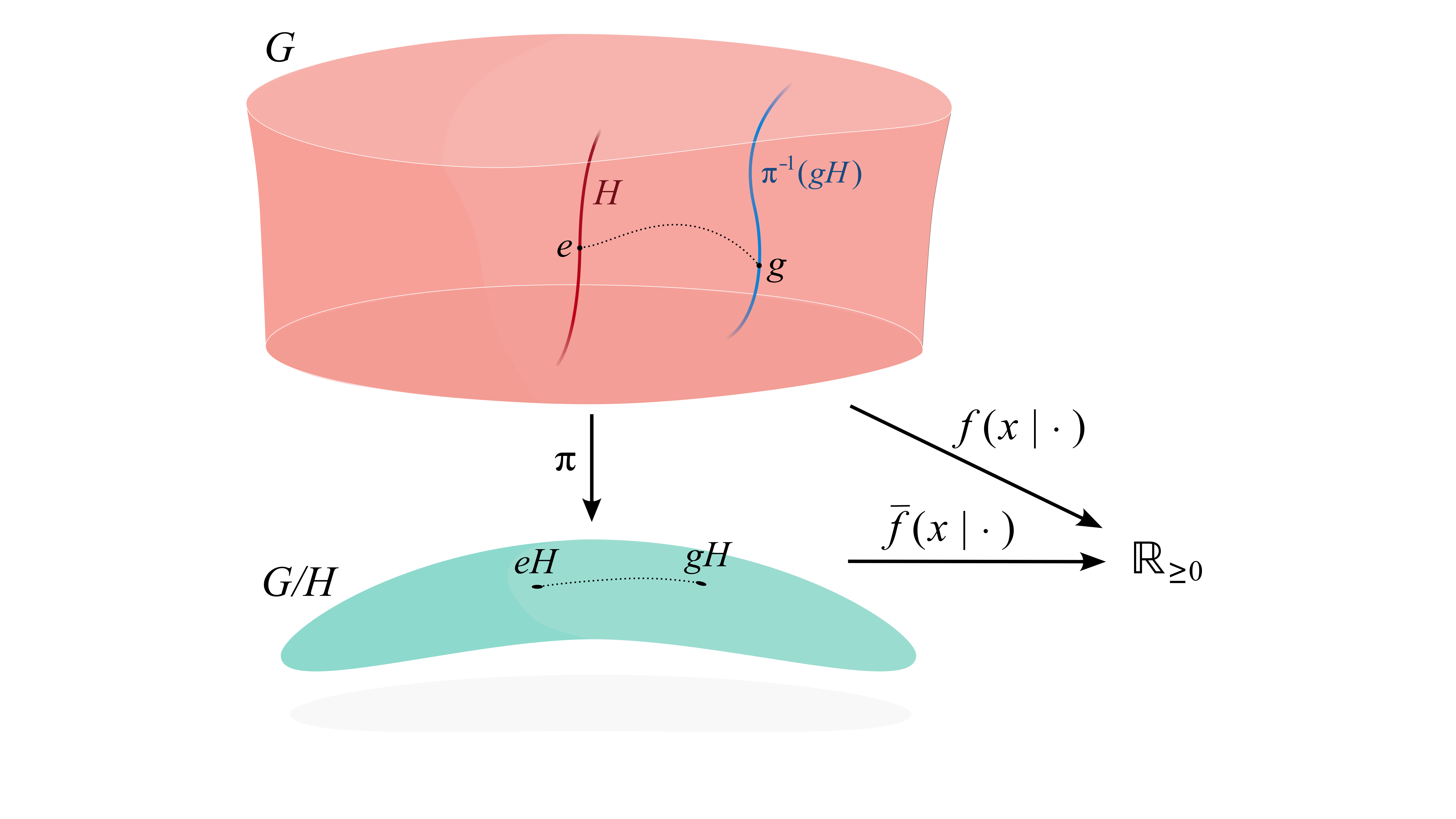}
    \caption{The map $\pi$ is a `projection' from $G$ to $G/H$ which maps $g\in G$ to its coset $\pi(g)=gH$. The statistical model $\bar f$ can be \textit{pulled back} via $\pi$ to define $f$ via (\ref{eq:pullback}). The inverse image under $\pi$ of the point $gH \in G/H$ is called the \textit{fiber} on $gH$.}
    \label{fig:homogeneous_space}
\end{figure}

\subsubsection{Tangent Spaces}
Let $\mathfrak g$ denote the Lie algebra of $G$, and $\mathfrak h \subset \mathfrak{g}$ the Lie subalgebra corresponding to the subgroup $H \subset G$.  
Let $d\pi_e$ denote the \textit{differential} (or \textit{pushforward map}) of $\pi$ at $e$. That is, given that $X\in \mathfrak g$ is the initial velocity of the curve $\exp(tX)$ on $G$, 
\begin{align}
    d\pi_e(X) = \frac{d}{dt} \pi\big( \exp(tX)\big)\bigg|_{t=0} = \frac{d}{dt} \exp(tX)H\bigg|_{t=0}
\end{align}
is the initial velocity of the image of the curve under $\pi$, which lies in $G/H$ and passes through $eH$ at $t=0$. Observe that $d\pi_e(Y)=0$ for all $Y\in \mathfrak h$.

In this article, we assume that $G/H$ is \textit{reductive} \cite[Def. 23.8]{gallier2020differential}, which means that $\mathfrak g$ can be decomposed as a direct sum $\mathfrak g = \mathfrak h \oplus \mf m$, where $\mf m$ is a subspace that is invariant under the adjoint action of $H$: $\Ad_h(\mf m)\subseteq \mf m$ for all $h\in H$. Note that each of the homogeneous spaces mentioned thus far in the article is reductive.
For example, an invertible matrix $\mf A \in GL(n)$ can be decomposed into its skew-symmetric and symmetric parts, as $\mf A = \tfrac{1}{2}(\mf A - \mf A^\top) + \tfrac{1}{2}(\mf A + \mf A^\top)$, which is a reductive decomposition for $\mathbb S(n)^{++}$.

\subsubsection{Invariant Vector Fields}
Given a vector $X\in \mathfrak g$, we let $X^L$ denote the left-invariant vector field (LIVF) generated by $X$, i.e., $X^L_g = (dL_g)_e(X) = gX$ for all $g\in G$. Similarly, $X^R_g=Xg$ is the corresponding right-invariant vector field (RIVF) evaluated at $g\in G$. The flows of LIVFs (resp., RIVFs) are given by right (resp., left) translations, so that each of them defines a derivative operator at $g$:
\begin{align}
    \left(X^L f\right)(g) &= \frac{d}{dt} f(g\exp(tX))\bigg|_{t=0} \label{eq:left-derivative}\\ 
    \left(X^R f\right)(g) &= \frac{d}{dt} f(\exp(tX)g)\bigg|_{t=0}
\end{align}
where $f$ is a smooth function on $G$. \footnote{Note that the derivative operator corresponding to $X^L$ is denoted as $\tilde X^r$ in \cite{chirikjian2011stochastic}, with the letter $r$ signifying that it involves right translation of $g$, i.e., `$\exp(tX)$' appears on the right of $g$ in (\ref{eq:left-derivative}) \cite{chirikjian2011stochastic}. Our notation and terminology follow those of \cite[Sec. 19.2]{gallier2020differential} instead.} The LIVFs and RIVFs are related by the \textit{adjoint action} of $G$ on $\mathfrak g$ \cite{chirikjian2011stochastic}, as follows:
\begin{align}
    X^R_g = (\Ad_{g^{-1}} X)^L_g \quad \text{and} \quad X^L_g = (\Ad_{g} X)^R_g.
    \label{eq:LIVF-RIVF}
\end{align}

\subsubsection{Bases}
Let $n_G$, $n_H$ and $n_{\Theta}$ be the dimensions of $\mathfrak g$, $\mathfrak h$ and $\mf m$ (or equivalently, those of $G$, $H$, and $G/H$), respectively, satisfying $n_G=n_H+n_{\Theta}$.
We choose a basis $\lbrace E_i\rbrace_{i=1}^{n_G}$ of $\mathfrak g$ such that $\lbrace E_i\rbrace_{i=1}^{n_H}$ is a basis of $\mathfrak h$ and $\lbrace E_i\rbrace_{n_H+1}^{n_G}$ a basis of $\mf m$. 
It is readily verified that the following vectors are tangent to the fiber on $gH$ (i.e., the vertical strand through $g$ in Fig. \ref{fig:homogeneous_space}):
\begin{align}
\big\lbrace E_{i,g}^L\big\rbrace_{i=1}^{n_H} \textit{ is a basis for }  \ker \left(d\pi_g\right) 
\end{align}
whereas its complement provides a basis for the tangent space of $G/H$ at $gH$:
\begin{align}
    \big\lbrace d\pi_g \left(E_{i,g}^L \right) \big\rbrace_{i=n_H+1}^{n_G} \textit{ is a basis for $T_{gH}(G/H)$}.
    \label{eq:LIVF_basis_on_GH}
\end{align}
In the above, $E_{i,g}^L=g E_i$, i.e., the LIVF generated by $E_i$ evaluated at $g$. 

Lastly, the map $(\,\cdot\,)^\vee: \mathfrak g \rightarrow \mathbb R^{n_G}$ is used to identify a vector in $\mathfrak g$ with its components in the basis $\lbrace E_i \rbrace_{i=1}^{n_G}$, while $(\,\cdot\,)^\wedge$ denotes the inverse of $(\,\cdot\,)^\vee$. The choice of basis determines an inner product (and vice versa), since for $X_1, X_2 \in \mathfrak g$, we may define $\langle X_1, X_2 \rangle \coloneqq ({X_1^\vee})^\top X_2^\vee$. 


\section{Fisher Information Metric (FIM)}\label{sec:FIM}
Recall that $f = \bar f \circ \pi$ is the pullback likelihood function. 
We use the following shorthand notation:
$$\bar \ell \coloneq \log \bar f\quad \text{and}\quad \ell\coloneq \log f$$
which represent the log-likelihood function and its pullback, respectively.
Given two vectors $V_1, V_2 \in T_g G$,
the \textit{Fisher Information Metric (FIM)}\footnote{Note that $F$ is technically not a Riemannian metric on $G$. Since the FIM of the pullback statistical model $f$ on $G$ is degenerate, it gives $G$ the structure of a \textit{sub-Riemannian manifold} \cite{strichartz1986sub}.
} at $g$ is defined as follows:
\begin{align}
F(V_1,V_2)_g &\coloneq \int_{\mathcal X} \big(V_1\ell(g) \; V_2 \ell(g) \big) f(g) \, d\mu \nonumber\\
&= \mathbb E \left[ V_1 \ell (g)\,  V_2\ell(g)\right]
\label{eq:fim}
\end{align}
where $\mathbb E$ denotes the expectation taken w.r.t. the density $f(g) = f(\,\cdot\,|g)$ and $V_1 \ell$ is the action of $V_1$ on $\ell$ as a derivative operator \cite{smith2005covariance}. The FIM at $g$ can be represented as a matrix after choosing the basis $\lbrace E_{i,g}^L \rbrace_{i=1}^{n_G}$ for $T_g G$, giving us the \textit{left Fisher Information Matrix} (left {\bf FIM}) at $g$. Its $(i,j)$-th entry is given by
\begin{align}
    \mf F_{ij,g}^L 
    &\coloneq F\big(E_{i,g}^L,\, E_{j,g}^L\big)_g 
    = \mathbb E \left[ E_i^L\ell(g)\, E_j^L \ell(g) \right].
    \label{eq:fim_matrix_left}
\end{align}
An analogous definition involving the RIVFs gives us $\mf F_{g}^R$, the right {\bf FIM} at $g$. Clearly, the two matrices $\mf F_{g}^L$ and $\mf F_{g}^R$ encode the same information; they both describe the metric defined in (\ref{eq:fim}), albeit using different bases.\footnote{This is identical to how, given a choice of basis at some point of a Riemannian manifold, the Riemannian metric is completely described by a matrix of coefficients known as the \textit{metric tensor}.}
When $G$ is abelian (i.e., commutative), we have $\mf F_{g}^L = \mf F_{g}^R$.


Our objective is to study the {FIM} of the original statistical model $\bar f$ on $G/H$. Given two vectors $\bar V_1, \bar V_2 \in T_{gH}G/H$, the FIM at $gH$ is
\begin{align}
    \bar F(\bar V_1,\bar V_2)_{gH} &\coloneq \mathbb E \left[ \bar V_1 \bar \ell (gH)\, \bar V_2\bar \ell(gH)\right].
    \label{eq:fim_on_GH}
\end{align}
We can define the corresponding {\bf FIM} using the basis given in (\ref{eq:LIVF_basis_on_GH}):
\begin{align}
    \overline{\mf F}_{ij,g} &\coloneq \bar {F}\Big(\,
    d\pi_g\left(E_{i',g}^L\right)\,,\;d\pi_g\left(E_{j',g}^L\right)
    \,\Big)_{gH},
    \label{eq:fim_matrix_on_GH}
\end{align}
where $i' = n_H+i$ and $j' = n_H+j$. 
When the FIM given in (\ref{eq:fim_on_GH}) is non-degenerate (or equivalently, when the {\bf FIM} in (\ref{eq:fim_matrix_on_GH}) is non-singular), it defines an inner product on $T_{gH}(G/H)$ as well as a
 Riemannian metric\footnote{Succinctly, a Riemannian metric is a `smoothly varying' inner product on the tangent spaces of a manifold. The smoothness of the FIM follows from our assumption that the statistical model $\bar f$ is smooth.} on $G/H$. 
 In this way, it endows $G/H$ with a unique geometric structure known as the \textit{Fisher-Rao geometry} \cite{liang2019fisher,amari1998natural}.

The relationships between the matrices, $\mf F_{g}^L$, $\mf F_{g}^R$, and $\overline{\mf F}_{g}$, are made explicit in the following lemma.

\begin{lemma}[Properties of the {\bf FIM}s]
    The following are true:
    \begin{enumerate}
        \item The left {\bf FIM}, viewed as a block matrix, is of the form
        $$
        \mf F_{g}^L = \begin{bmatrix}
        \mf 0 & \mf 0 \\
        \mf 0 & \overline{\mf F}_{g}
        \end{bmatrix}
        $$
        \item The right {\bf FIM} is constant along the fibers:
        $$\mf F_{g}^R = \mf F_{gh}^R\ \forall h\in H $$
        \item The left and right {\bf FIM}s are related as follows:
        \begin{align}
            \mf F_{g}^L = {\boldsymbol \Ad}_{g}^\top \mf F_{g}^R {\boldsymbol \Ad}_{g}
        \end{align}
        where ${\boldsymbol \Ad}_{g}$ denotes $\Ad_g$ written as a matrix w.r.t. the basis $\lbrace E_i \rbrace_{i=1}^{n_G}$.
    \end{enumerate}
    \label{lem:fim-properties}
\end{lemma}
\begin{proof}
    Using (\ref{eq:fim_matrix_left}), we have
    \begin{align*}
        \mf F_{ij,g}^L
        = \mathbb E \left[\frac{d}{dt} \ell\big(g \exp(t E_i)\big)\bigg|_{t=0}\, \frac{d}{dt} \ell\big(g \exp(t E_j)\big)\bigg|_{t=0} \right]
    \end{align*}
    When $i \leq n_H$, we have $E_i \in \mathfrak h$, so that 
    $$\ell\big(g \exp(t E_i)\big) = \ell\big(g h(t)\big) = \ell(g)$$ 
    for some curve $h(t)\in H.$ 
    Thus, $\mf F_{ij,g}^L$ is zero when either $i$ or $j$ is less than $n_H$.
    Next, we focus on the matrix $\overline {\mf F}_g$. Observe that
    \begin{align}
        \left[d\pi_g\, E_{i,g}^L\right]\bar \ell(gH)= E_{i,g}^L\left(\bar \ell\circ \pi \right)(g) = E_{i,g}^L\ell(g)
    \end{align}
    which follows from the definition of the differential. 
    By substituting the above in (\ref{eq:fim_matrix_on_GH}), we see that $\overline {\mf F}_g$ is a sub-matrix of $\mf F_g^L$, which proves property 1. Property 2 is proved similarly.

    To show property 3, let $\mf x = [x^1\ x^2\ \cdots\ x^{n_G}]^\top$ be an arbitrary vector in $\mathbb R^{n_G}$ and $X \coloneq \mf x^\wedge$. Using the Einstein summation convention, we have $X^L = x^i E_i^L$. 
    Since $F$ is linear in either argument, we have
    \begin{align}
        \mf x^\top \mf F_{g}^L \,\mf x &= F(x^iE_{i,g}^L, x^jE_{j,g}^L)_g = F(X_g^L, X_g^L).
    \end{align}
    On the other hand, using (\ref{eq:LIVF-RIVF}), we get 
     \begin{align}
        \mf x^\top {\boldsymbol \Ad}_{g}^\top \mf F_{g}^R {\boldsymbol \Ad}_{g} \mf x &= F\left((\Ad_g X)_g^R, (\Ad_g X)_g^R\right)\\
        &= F\left(X^L_g, X^L_g\right)
    \end{align}
    which concludes the proof.
\end{proof}
\noindent
Combining properties 2 and 3, we obtain
\begin{align}
 \mf F_{gh}^L = \bm\Ad_{h}^\top \mf F_{g}^L \bm\Ad_{h},
 \label{eq:fimL-gh}
\end{align}
i.e., the left {\bf FIM} is \textit{not} constant along the fibers in general. Evidently, the left and right {\bf FIM}s of $G$ have complementary properties.

\section{Cram\'er-Rao Bound (CRB)}\label{sec:CRB}
\subsection{The CRB on Lie Groups}\label{sec:CRB-on-G}
Let $\hat g: \mathcal X \rightarrow G$ be an \textit{estimator} on $G$ which maps a given observation in $\mathcal X$ to a unique point in $G$.
The \textit{left-invariant error} of $\hat g$ w.r.t. $g$ is defined as $\eta(g) \coloneq \log (g^{-1} \hat g)$ \cite{barrau2016invariant}.
We define the (left) \textit{bias vector} of $\hat g$ as
\begin{align}
    B(g) \coloneq \mathbb E [\,\eta (g)\,] = \mathbb E [\,\log (g^{-1}\hat g)\,] \in \mathfrak g,
    \label{eq:bias}
\end{align}
where $\log(\,\cdot\,)$ denotes the log map of $G$, i.e., the matrix logarithm.\footnote{We assume that the $\log$ operation in (\ref{eq:bias}) is well-defined. For e.g., in the case of $SO(3)$, $\log$ is well-defined everywhere but on a negligible (measure-zero) subset \cite{chirikjian2011stochastic}.} The error and bias can be represented using the vectors $\bm \upeta(g) \coloneq \eta(g)^\vee$ and $\mf b(g) \coloneq B(g)^\vee$, and the (left) \textit{covariance matrix} of $\hat g$ may be defined in analogy with the Euclidean case:
\begin{align}
    \mf\Sigma(g) = \mathbb E \left[\big( \bm\upeta(g) - \mf b(g) \big)\big( \bm\upeta(g) - \mf b(g) \big)^\top\right].
    \label{eq:covariance}
\end{align}
The adjoint action is used to convert the left-invariant error to the right-invariant error: $\bm\upeta'(g) \coloneq \bm\Ad_g \bm\upeta(g) = \log( \hat g g^{-1})^\vee$. Let $\mf b'(g)$ and $\mf\Sigma'(g)$ denote the corresponding (right-invariant) bias and covariance, respectively.

We define $\mf J_{\mf b}(g)$ as the Jacobian of the map $\mf b: G \rightarrow \mathbb R^n$ whose components are $\left[\mf J_{\mf b}(g)\right]_{ij} = E_j^Lb^i(g)$, with $b^i(g)$ representing the $i^{th}$ component of $\mf b(g)$. Similarly, we define the Jacobian of the right bias as $\mf J_{\mf b'}(g)$, where $\left[\mf J_{\mf b'}(g)\right]_{ij}=E_j^R{b'}^i(g)$. We can now state the group-theoretic CRB for biased estimators on $G$.

\begin{theorem}[CRB on Lie Groups]
    Let $\hat g$ be an estimator on $G$.
    The left and right covariances of $\hat g$ (i.e., $\mf \Sigma(g)$ and $\mf \Sigma'(g)$, respectively) satisfy the following inequalities:
    \begin{align}
         \bm \Phi(g){\mf F^L_g}^\dagger \,\bm \Phi(g)^\top 
    \ &\;\preccurlyeq\;  \mf \Sigma(g)\label{eq:left-crb}\\
    \bm \Phi'(g){\mf F^R_g}^\dagger \,\bm \Phi'(g)^\top
       &\;\preccurlyeq\; \mf \Sigma'(g)
        \label{eq:right-crb}
    \end{align}
    where $\bm \Phi(g) \coloneq \mathbb E \big[\boldsymbol\Psi_{-\eta(g)}\big] + \mf J_{\mf b}(g)$, $\,
        \bm \Phi'(g) \coloneq \mathbb E \big[\boldsymbol\Psi_{\eta'(g)}\big] + \mf J_{\mf b'}(g)$,
        $\;\mf \Psi_{(\,\cdot\,)}$ is the matrix defined in Lemma \ref{lem:derivative-log} in Appendix \ref{app:proof-crb-group}, and $(\,\cdot\,)^\dagger$ denotes the pseudoinverse.
    \label{thm:crb-group}
\end{theorem}
\begin{proof}
See Appendix \ref{app:proof-crb-group}.
\end{proof}
\noindent
The bias and variance of $\hat g$ may not be small; since $f(g) = f(gh)$, the observed measurement does not help us distinguish between $g$ and $gh$ (that is, unless $H=\lbrace e \rbrace$ is the trivial subgroup). 

\subsection{The CRB on Homogeneous Spaces}
\label{sec:CRB-GH}
Any estimator $\hat g$ on $G$ \textit{induces} an estimator $\hat \theta$ on $G/H$: 
$$\hat \theta(x) \coloneqq \pi \circ \hat g\,(x) = \hat g(x) H.$$
 For example, if our goal is to estimate a point on the sphere $S^2$ based on the observation $x\in\mathcal X$, we may instead estimate a rotation matrix $\hat{\mf R} (x)\in SO(3)$, so that $\pi\big(\hat {\mf R}(x)\big)$ simply extracts the first column of $\hat {\mf R}(x)$, i.e., a point on the sphere. The covariance of $\hat {\mf R}$ satisfies the CRB given in Theorem \ref{thm:crb-group}. We now investigate the converse direction; suppose we are given an estimator $\hat \theta$ on $G/H$. Does this correspond to a unique estimator on $G$? How can we state the CRB for $\hat \theta$?

If the estimator $\hat \theta$ is sufficiently concentrated near some parameter $gH$, 
then we can convert it into a unique estimator $\hat g$ near $g$ (where $g$ is a given representative of $gH$) using the following procedure. 
Consider the submanifold $$\mathcal S_g \coloneq \lbrace g \exp Y | Y\in \mf m, \lVert Y \rVert \leq \epsilon \rbrace \subseteq G$$
where $\lVert \,\cdot\,\rVert $ is an arbitrary norm on $\mf m$.
For sufficiently small $\epsilon$, $\mathcal S_g$ is a \textit{horizontal cross-section} of $G$ which is diffeomorphic to (i.e., looks like) the subset $\pi(\mathcal S_g)$ of $G/H$. Consequently, there is a local right-inverse of $\pi$, denoted as $\sigma_g : \pi(\mathcal S_g) \rightarrow \mathcal S_g$, such that $\pi \circ \sigma_g$ is the identity map on $\mathcal S_g$.
Assuming that the range of $\hat \theta$ is contained in $\pi(\mathcal S_g)$, we can define the estimator $\hat g \coloneq \sigma(\hat \theta): \mathcal X \rightarrow \mathcal S_g$, which we call as the \textit{horizontal lift} of $\hat \theta$ at $g$. In this sense, there is a one-to-one correspondence between estimators on $G/H$ and estimators on $\mathcal S_g$ for each $g\in G$. 

When the CRB on Lie groups is applied to a horizontally lifted estimator on $G$, we obtain the CRB on $G/H$.
\begin{theorem}[CRB on Homogeneous Spaces]
    Let $\hat \theta$ be an estimator on $G/H$ and $\hat g$ its horizontal lift at $g$. Then, the left-invariant error, bias, and covariance of $\hat g$ are of the form
    \begin{align*}
        \bm \upeta(g) &= \begin{bmatrix}
        \bm 0\\
        \bar{\bm \upeta}(g)
        \end{bmatrix}, \  \mf b(g) = \begin{bmatrix}
        \mf 0\\
        \bar{\mf b}(g)
        \end{bmatrix}, \;\textit{and }\  \mf \Sigma(g) = \begin{bmatrix}
        \mf 0 & \mf 0\\
        \mf 0 & \overline{\mf \Sigma}(g)
        \end{bmatrix},
    \end{align*}
    respectively. 
    Furthermore, the following inequality holds:
    \begin{align}
        \left(\bm \Pi \bm \Phi(g) \bm \Pi^\top\right)\,\overline {\mf F}_g^{-1} \left(\bm \Pi \bm \Phi(g) \bm \Pi^\top\right)^\top
    \;\preccurlyeq\; \overline {\mf\Sigma}(g)
    \label{eq:crb-GH}
    \end{align}
    where $\bm \Pi \coloneq \begin{bmatrix}\, \bm 0 & \mf I \,\,\end{bmatrix}$.\footnote{Here, the dimensions of $\mf 0$ and $\mf I$ are inferred from the context: $\mf 0$ is the ${n_\Theta \times n_H}$ matrix of zeros and $\mf I$ is the ${n_\Theta \times n_\Theta}$ identity matrix.}
    \label{thm:crb-group-homogeneous}
\end{theorem}
\begin{proof}
    Observe that $\hat g$ is of the form $g \exp \hat Y$, where $\hat Y:\mathcal X \rightarrow \mf m$ is an $\mf m$-valued estimator. The left-invariant error becomes
    $$
    \eta(g) = \log(g^{-1} \hat g) = \log\big(\exp(\hat Y)\big) = \hat Y.
    $$
    Hence, $\eta(g) \in \mf m$.
    Due to our choice of basis (cf. Section \ref{sec:GH-structure}), the first $n_H$ components of $\bm \upeta(g)=\eta(g)^\vee$ are zero, which establishes the first part of the theorem.

    To show the CRB inequality, we note that the pseudoinverse of $\mf F_g^L$ is given by
    \begin{align}
        {\mf F^L_g}^\dagger = \begin{bmatrix}
        \mf 0 & \mf 0 \\
        \mf 0 & \overline {\mf F}(g)^{-1}
        \end{bmatrix} = \mf \Pi^\top \overline {\mf F}(g)^{-1} \mf \Pi.
    \end{align}
    Multiplying the left CRB (\ref{eq:left-crb}) by $\mf \Pi(\,\cdot\,)\mf \Pi^\top$ on either side gives us the desired inequality, while observing that
        $
        \mf \Pi \mf \Sigma(g) \mf \Pi^\top
        =\overline {\mf\Sigma}(g).
        $
\end{proof}
\noindent
    In Appendix \ref{app:relationship}, we show that $\bar{\bm \upeta}(g)$ is an appropriate definition for the estimation error of $\hat \theta$ since it agrees with the intrinsic (Riemannian) definition under certain conditions. Furthermore, it is easy to see that $\bar{\bm \upeta}(g) = \mf 0$ if and only if $\hat \theta = gH$. 
%
\noindent
\begin{corollary}[Third-Order CRB for Unbiased Estimators]
    Suppose $\hat \theta$ is unbiased,
    i.e., $\bar{\mf b}(g) = \mf 0$. Then, 
    neglecting the fourth-order terms in (\ref{eq:crb-GH}),
    we have
    \begin{align}
        \big(\mf I +  \mf \Delta(g)\big) \overline {\mf F}_g^{-1} \big(\mf I + \mf \Delta(g)\big)^\top
        \preccurlyeq \overline {\mf\Sigma}(g),
        \label{eq:crb-GH-unbiased}
    \end{align}
    where $\mf \Delta(g) \coloneq \mf \Pi\; \mathbb E\left[\frac{1}{12}\bm \ad_{\eta(g)}^2 \right] \mf \Pi^\top$ and $\bm \ad_{(\cdot)}$ denotes the matrix representation of $\ad_{(\cdot)}$.
    \label{cor:crb-GH-unbiased}
\end{corollary}
\begin{proof}
Using the definitions in Lemma \ref{lem:derivative-log} and Remark \ref{rem:bernoulli} of Appendix \ref{app:proof-crb-group}, we get 
\begin{align}
\boldsymbol\Psi_{-\eta(g)} = \mf I - \frac{1}{2}\bm \ad_{\eta(g)} + \frac{1}{12}\bm \ad_{\eta(g)}^2 + O(\lVert \bm\upeta(g)\rVert^4),
\end{align}
whereas $\mf J_{\mf b}(g) = \mf 0$ due to unbiasedness. Hence, we have
\begin{align}
    \bm \Phi(g) = \mathbb E \left[\boldsymbol\Psi_{-\eta(g)} \right] &\approx  \mf I + \mathbb E\left[-\frac{1}{2}\bm \ad_{\eta(g)} + \frac{1}{12}\bm \ad_{\eta(g)}^2 \right]\\
    &= \mf I +  \mathbb E\left[\frac{1}{12}\bm \ad_{\eta(g)}^2 \right],
    \label{eq:phi-approx}
\end{align}
since $\ad_{(\cdot)}$ is linear in $\eta (g)$ and $\mathbb E[\eta(g)] = 0$ for the horizontal lift of an unbiased estimator. Substituting into (\ref{eq:crb-GH}) gives us (\ref{eq:crb-GH-unbiased}).
\end{proof}

An estimator is said to be \textit{efficient} if it achieves the CRB, i.e., if the inequality in (\ref{eq:crb-GH}) is an equality. While an efficient estimator does not necessarily exist, a sufficient and necessary condition for efficiency is given in the next corollary. It is obtained using the fact that the Cauchy-Schwarz inequality used in the proof of Theorem \ref{thm:crb-group} becomes an \textit{equality} when the corresponding vectors are linearly dependent (c.f. \cite[Cor. 3]{smith2005covariance}).

\begin{corollary}[Efficiency]
    The estimator $\hat \theta$ is \textit{efficient} (i.e., achieves the lower bound given in Theorem \ref{thm:crb-group-homogeneous}) if and only if its horizontal lift $\hat g$ satisfies
    \begin{align}
        &  {\bm \upeta}(g) = {\mf b}(g)
   + c \cdot  \bm \Phi(g) \bm \Pi ^\top \overline {\mf F}_g^{-1} \,\overline{\mf{grad}}\,\ell(g),
   \label{eq:efficient-estimator}
    \end{align}
    where $c \in \mathbb R$ is a scalar and $$\overline{\mf{grad}}\, \ell(g) \coloneq \left[E_{n_H+1}^L \ell(g)\ \cdots E_{n_G}^L \ell(g)\right]^\top$$ represents the gradient of $\ell$ w.r.t. the LIVFs generated by $\mf m$.
    \label{cor:efficiency}
\end{corollary}
Using (\ref{eq:phi-approx}) and (\ref{eq:efficient-estimator}), we see that an unbiased efficient estimator satisfies $
    {\bm \upeta}(g) \propto  \bm \Pi ^\top \overline {\mf F}_g^{-1} \,\overline{\mf{grad}}\,\ell(g) $
up to an $O(\bm \Delta(g))$ error. Note that either side of this equation is treated as a deterministic vector-valued function on $\mathcal X$, since $\hat g(\,\cdot\,)$ and $\ell(\,\cdot\,|\,g)$ are each functions of the observation $x\in \mathcal X$. 


\begin{remark}
Theorem \ref{thm:crb-group-homogeneous} and its corollaries are also applicable to group-valued estimators. To do this, we view a group $G$ as the homogeneous space $G/\lbrace e \rbrace$ while making the substitutions $\bm \Pi \mapsto \mf I$ and $n_H \mapsto 0$.
In this sense, Corollary \ref{cor:crb-GH-unbiased} is consistent with the result for Lie groups derived in \cite{bonnabel2015intrinsic}.
\end{remark}

\subsection{Bound on the Variance}\label{sec:variance}
While a bound on the variance is readily obtained by taking the trace of (\ref{eq:crb-GH}), the bound is particularly interesting for a special class of homogeneous spaces, namely those that admit a \textit{$G$-invariant metric}. 
Roughly speaking, if $G/H$ is endowed with a $G$-invariant metric, then it retains its shape when operated upon by $G$; for e.g., the unit sphere in $\mathbb R^3$ looks identical when it is rotated by an element of $SO(3)$.
Unlike left-invariant metrics on $G$ (which always exist), $G/H$ admits
a $G$-invariant metric if and only if there exists an {$\Ad_H$-invariant inner product} on $\mf m$, i.e., an inner product satisfying
\begin{align}
    \langle Z_1, Z_2 \rangle = \langle \Ad_h Z_1, \Ad_h Z_2 \rangle
    \label{eq:Ad-invariance}
\end{align}
for all $Z_1, Z_2 \in \mf m$ and $h\in H$. A sufficient condition for such an inner product to exist is for $H$ to be compact \cite[Prop. 3.16]{cheeger1975comparison} \cite[Prop. 23.22]{gallier2020differential}.

In the forthcoming lemma, we will suppose that the basis $\lbrace E_i \rbrace_{i=1}^{n_G}$ introduced in Section \ref{sec:GH-structure} is orthonormal w.r.t. the choice of inner product on $\mathfrak g$. 
Note that we have assumed implicitly that $\mathfrak h$ and $\mf m$ are mutually orthogonal. 
Moreover, due to our assumption that $\Ad_H(\mf m) \subseteq \mf m$ in Section \ref{sec:GH-structure}, the subspaces $\mathfrak h$ and $\mf m$ are each invariant under $\Ad_h$ for all $h\in H$. Thus, $\bm \Ad_h$ takes the form
\begin{align}
    \bm \Ad_h = \begin{bmatrix}
    \bm \Ad_h\rvert_{\mathfrak h} & \bm 0\\
    \bm 0 & \bm \Ad_h\rvert_{\mf m}
    \end{bmatrix}
\end{align}
and condition (\ref{eq:Ad-invariance}) is equivalent to $\bm \Ad_h\rvert_{\mf m}$ being orthogonal.
\begin{lemma}
    Suppose the inner product on $\mathfrak g$ is $\Ad_H$-invariant when restricted to $\mf m$ (or equivalently, $\bm \Ad_h\rvert_{\mf m}$ is orthogonal). Then, for all $h\in H$, we have
    \begin{align*}\tr \left(\overline {\mf \Sigma}(g) \right) = \tr \left( \overline {\mf \Sigma}(gh) \right)
    \ \textit{ and }\   \tr \big( \overline {\mf F}_g^{-1} \big) = \tr \big( \overline {\mf F}_{gh}^{-1} \big).\end{align*}
    \label{lem:variance}
\end{lemma}
\begin{proof}
    Let $\hat \theta$ be an estimator on $G/H$ and $ \hat g \coloneq g \exp \hat Y$ its horizontal lift at $g$. The horizontal lift of $\hat \theta$ at $gh$ is then $\hat g_0 \coloneq gh \exp (\Ad_{h^{-1}}\hat Y)$. To see this, we recall from Section \ref{sec:GH-structure} that $\hat Y \in \mf m \Rightarrow \Ad_{h^{-1}}\hat Y \in \mf m$ so that $\hat g_0$ is indeed a horizontal lift. We also need to verify that $\pi \circ \hat g_0 = \pi \circ \hat g$:
    \begin{align}
        \pi \circ \hat g_0 &= 
        gh \exp (\Ad_{h^{-1}}\hat Y)H= gh h^{-1}\exp(\hat Y) h H \\
        &= g\exp(\hat Y) H = \pi \circ \hat g.
    \end{align}

    Let $\eta(g)$ and $\eta_0(gh)$ be the left-invariant errors of $\hat g$ and $\hat g_0$. Then, we have
    \begin{align}
        \lVert \eta(g) \rVert^2 &= \lVert \hat Y \rVert^2 = \lVert \Ad_{h^{-1}}\hat Y \rVert^2 = \lVert \eta_0(gh) \rVert^2
    \end{align}
    due to the $\Ad_H$-invariance of the norm. Noting that $\tr(\overline {\mf \Sigma}(g)) = \mathbb E\left[\lVert {\bm \upeta}(g) \rVert^2\right]$ completes the proof of the first part.
    
    To show the second part, we use Lemma \ref{lem:fim-properties} and (\ref{eq:fimL-gh}):
    \begin{align}
        \tr\big(\overline {\mf F}_{gh}^{-1}\big) 
        &= \tr\big(\bm \Pi \mf F_{gh}^L \bm \Pi^\top\big) 
        = \tr\big(\bm \Pi \bm \Ad_{h}^\top \mf F_{g}^L \bm \Ad_{h} \bm \Pi^\top\big) \nonumber\\
        &=\tr\left(\bm \Pi \bm \Ad_{h}^\top 
            \big(\bm \Pi^\top \overline {\mf F}_{g}^{-1}\bm \Pi \big) \bm
            \Ad_{h} \bm \Pi^\top\right) \nonumber\\
            &= \tr\big(\bm \Ad_h\rvert_{\mf m}^\top \ \overline {\mf F}_{g}^{-1}\,\bm \Ad_h\rvert_{\mf m}\big) = \tr\big(\overline {\mf F}_{g}^{-1}\big)
    \end{align}
    since $\bm \Ad_h\rvert_{\mf m} \bm \Ad_h\rvert_{\mf m}^\top = \mf I$.
\end{proof}

Lemma \ref{lem:variance} shows that there is a consistent definition of the \textit{variance} of $\hat \theta$ at $gH$ that is independent of the choice of representative $g$, i.e., $\var(gH) \coloneq \tr \left( \overline {\mf \Sigma}(g) \right) = \mathbb E[\lVert \bar {\bm \upeta}(g) \rVert^2]$. The intuition behind Lemma \ref{lem:variance} is that the basis we defined in (\ref{eq:LIVF_basis_on_GH}) {rotates} (but does not shrink or enlarge) when we move from the point $g$ to $gh$ on the same fiber. Thus, error vector $\bar{\bm \upeta}(g)$ depends on the choice of representative $g$, its length does not.

Taking the trace of (\ref{eq:crb-GH}), we get
$\tr\big(\overline {\mf F}_g^{-1}\big) 
        \preccurlyeq \var (gH),$
which is a bound on the variance of an unbiased estimator that is accurate up to an 
$O({\bm \Delta}(g))$ error. Importantly, this bound is independent of the choice of representative $g$ of $gH$, provided that the condition in Lemma \ref{lem:variance} holds.

\subsection{{Parameter Estimation} on $H\backslash G$}\label{sec:HG}
In certain applications, it is more natural to stipulate that the group $H$ acts on the left, leading to a parameter estimation problem on the right coset space $H \backslash G \coloneq \lbrace Hg \, | \, g\in G\rbrace$. In this case, the properties of the left and right {\bf FIM}s in Lemma \ref{lem:fim-properties} are interchanged, and therefore $\overline {\mf F}_g$ is defined as a sub-matrix of ${\mf F}_g^R$ instead of ${\mf F}_g^L$. Similarly, the roles of LIVFs and RIVFs are reversed; the left-invariant error $\eta(g)$ is replaced by the right invariant error $\eta'(g)$, while the gradient vector used in (\ref{eq:efficient-estimator}) is redefined w.r.t. the RIVFs instead of LIVFs.

\section{Applications}\label{sec:applications}
In this section, we discuss several promising applications of the theory developed in the preceding sections.

\subsection{Computation of the FIM and CRB}
If one has access to the pdf $\bar f(\theta)$, then the {\bf FIM} may be computed analytically using the definitions in Section \ref{sec:FIM}.
An alternative way to compute the {\bf FIM} is by relating it to the Hessian of $\ell(g)$:
\begin{align}\mf F_{ij,g}^L = - \mathbb E \left[E_j^L E_i^L \ell(g)\right]=- \mathbb E \left[E_i^L E_j^L \ell(g)\right].
\label{eq:fim-hessian}
\end{align}
The derivation of this identity follows similarly to the Riemannian case \cite{bonnabel2015intrinsic}, and is presented in the supplementary document. 
Recently, the authors of \cite{habi2023learning} proposed a way to use generative models to estimate the FIM and CRB. This is done by using a dataset of observation-parameter pairs $\lbrace(x_1, \theta_1), (x_2, \theta_2), \cdots, (x_n, \theta_n)\rbrace$ to learn the pdf $\bar f(x|\theta)$, after which standard auto-differentiation toolboxes are used to estimate the FIM. Empirical estimation of the FIM is also studied in the deep learning literature \cite{zhang2019fast,karakida2020understanding,amari1998natural,martens2020new,faria2023fisher}.

Oftentimes, one is able to modify the pdf $\bar{f}(\theta)$ (which in turn modifies the FIM and its associated CRB) directly. For instance, in the case where $x\in \mathcal X$ represents a vector of measurements obtained over a \textit{sensor network}, the pdf $\bar{f}(\theta)$ depends on the connectivity of the network as well as the pose of each sensor. Thus, one can consider minimizing the trace or determinant of the FIM to determine the optimal sensor configuration, subject to practical constraints such as a limited resource of sensors, desired region of coverage, etc. \cite{chen2021cramer}.

\subsection{Iterative Parameter Estimation}


Assume that $(x_1, x_2, \ldots, x_m)$ is a sequence of $m$ independent samples from the pdf $f(g)$, with $m>0$. The pdf of the observed measurement sequence is $\Pi_{i=1}^m f(x_i|g)$ and its logarithm is $\sum_{i=1}^m \ell(x_i|g)$. Due to (\ref{eq:fim-hessian}), the {\bf FIM} is the sum of the {\bf FIM}s of the individual measurements, i.e., $\sum_{i=1}^m \overline {\mf F}_{g} = m\,\overline{\mf F}_{g}$. Let $\hat g$ represent an unbiased estimator of $g$.
Using Corollary \ref{cor:efficiency}, we have 
\begin{align}
     \log(g^{-1}\hat g)^\vee \propto \frac{1}{m}  \bm \Pi^\top \overline {\mf F}_g^{-1} \sum_{i=1}^m\overline{\mf{grad}}\,\ell(x_i | g) 
     \label{eq:efficiency_multiplemeasurements}
\end{align}
up to an $O(\mf \Delta)$ error, where $\hat g$ is understood as $\hat g(x_1, \ldots, x_m)$. In some parameter estimation problems (e.g., least squares regression) one can solve (\ref{eq:efficiency_multiplemeasurements}) to derive the efficient estimator in closed-form. Alternatively, one sets the gradient vector $\overline{\mf{grad}}\,\ell$ to zero, yielding the {Maximum Likelihood Estimator (MLE)} of $g$.
When an analytical (closed-form) estimator is intractable, an iterative approach is sought.

One strategy is to consider the following iteration, which we call the \textit{generalized Fisher scoring} algorithm:\footnote{For parameter estimation on the \textit{right} coset space, $H \backslash G$, the exponential map in (\ref{eq:fisher-scoring}) is placed to the \textit{left} of $\hat g_n$; also see Section \ref{sec:HG}.}
\begin{align}
    \hat g_{k+1} = \hat g_{k}\exp\left(\bigg(\frac{1}{m}\bm \Pi^\top \overline {\mf F}_{\hat g_k}^{-1}\sum_{i=1}^m  \,\overline{\mf{grad}}\,\ell(x_i | \hat g_k)\bigg)^\wedge \right)
    \label{eq:fisher-scoring}
\end{align}
where $\hat g_0$ is an initial guess of $g$.
The iteration in (\ref{eq:fisher-scoring}) is derived by considering a \textit{gradient ascent} approach for maximizing $\ell(\,\cdot\,)$ while using the FIM as the Riemannian metric.\footnote{If $\mf G$ is the metric tensor of a Riemannian manifold at some point, then the covector with coefficients $\mf v^\top$ is said to be \textit{paired} with (or \textit{dual} to) the vector $\mf G^{-1} \mf v$. The operation $\mf v^\top \mapsto \mf G^{-1} \mf v$ 
is a generalization of the vector transpose. In (\ref{eq:fisher-scoring}), we use the {\bf FIM} as the metric tensor to convert the gradient covector into a vector; see \cite[p. 342]{lee2012smooth} for details.} In light of the relationship of the FIM to the Hessian, we can also view (\ref{eq:fisher-scoring}) as the second-order \textit{Newton-Raphson} method (as done in \cite[Sec. 9.6]{Li2019}).
Its fixed points (i.e., the points $\hat g_k \in G$ that satisfy $\hat g_{k+1} = \hat g_k$ when subjected to the iteration) include the MLE, and potentially other local extrema. The matrix $\bm \Pi^\top$ ensures that the estimator does not update along the fiber, while the positive definite matrix $\overline {\mf F}_{\hat g_k}^{-1}$ scales the gradient appropriately, respecting the inherent Fisher-Rao geometry of the parameter estimation problem. For this reason, (\ref{eq:fisher-scoring}) (or rather, the Euclidean version of it) is also called \textit{natural gradient ascent} in the machine learning literature \cite{amari1998natural,martens2020new}. Importantly, note that (\ref{eq:fisher-scoring}) does not require a step-size, a fact which we verify through numerical simulations in the next section.

When $\Theta$ is a vector space, it reduces to the well-known (Euclidean) Fisher scoring algorithm, which can be shown to achieve the CRB within a single update of the algorithm, in the limit $m \rightarrow \infty$ \cite[Thm. 9.4]{Li2019}. 
 


\subsection{Group-Invariant Parameter Estimation}
\label{sec:group-invariance}
Homogeneous spaces can be used to systematically study parameter estimation problems which are invariant w.r.t. a continuous group of symmetries, where the group of symmetries is represented by $H$. Examples of parameter estimation problems which exhibit the aforementioned \textit{group-invariance}\footnote{Note that this is different from the predominant use of the words \textit{invariance} and \textit{equivariance} in the deep learning literature, which refers to the action of a group on the \textit{observations} rather than on the parameters \cite{bronstein2017geometric}.} phenomenon include
\begin{itemize}
    \item State estimation of symmetric objects, where $G = SE(3)$ describes the object's pose and $H \cong SO(2)$ describes rotations about its axis of symmetry. The observed measurements (e.g., point clouds) are statistically invariant under rotations by $H$ \cite{merrill2022symmetry}.
    \item Embedding words into a vector space, where $H = O(d)$ is the group of orthogonal transformations. Since the cosine similarity between a pair of vectors is invariant under the action of $H$, two embeddings that are related by an $O(d)$ transformation are equivalent \cite{yin2018dimensionality,mahadevan2015reasoning}.
\end{itemize}

To further illustrate our point, we consider the following examples from engineering. These examples demonstrate how our theory can be applied to systematically study statistical models that are invariant w.r.t. the action of some group.

\subsubsection*{\textbf{Example 1.} Pose-Estimation using Landmarks}
Let $g\in SE(3)$ describe the \textit{pose} (i.e., position and orientation) of a mobile robot. We write $g = (\mf R, \mf p)$, such that $\mf R \in SO(3)$ and $\mf p \in \mathbb R^3$ are understood to refer to the rotational and translational components of $g$, respectively. To estimate its own pose, the robot uses a \textit{landmark}, whose position $\mf a \in \mathbb R^3$ is known in a global frame of reference. The location of the landmark is measured in the robot's local frame of reference as the vector $\mf x \in \mathcal X = \mathbb R^3$ which is subject to standard Gaussian noise, yielding the likelihood function
\begin{align}
    f(\mf x\,|\,g) \propto \exp\left(-\frac{1}{2}\left\lVert \mf x - \mf R^\top (\mf a  - \mf p) \right\rVert^2\right)
    \label{eq:likelihood_landmark}
\end{align}
with $\lVert \, \cdot \, \rVert$ denoting the Euclidean norm \cite{chirikjian2014gaussian}. The problem of estimating $\mf R$ and $\mf p$ based on the pdf (\ref{eq:likelihood_landmark}) can be viewed as a generalization of \textit{Wahba's problem}. In Wahba's problem, $\mf p$ is assumed to be known, so that only $\mf R \in SO(3)$ needs to be estimated. Wahba's problem was used to validate the Lie group CRBs in \cite{bonnabel2015intrinsic} and \cite{solo2020cramer}.

It may be verified that $f(\mf x|g) = f(\mf x|hg)$ for all $h\in H$, where $H\subseteq SE(3)$ is defined as follows:
$$H\coloneq\left\lbrace \big(\mf Q, (\mf I - \mf Q)\mf a\big)\,\big|\,\mf Q \in SO(3)\right\rbrace.$$
This is precisely the subgroup of $SE(3)$ comprising transformations that have the point $\mf a$ as their fixed point.
Thus, we can view the pose estimation problem w.r.t. a single landmark as a parameter estimation problem on $H \backslash SE(3)$. Observe that $H$ is three-dimensional, so that $n_H=3$ and $n_{\Theta}= n_G - n_H = 3$.

Consider arbitrary matrices $X_i, X_j \in\mathfrak{se}(3)$, where we write $X_i = (\mf \Omega, \mf v_i)$ with $\mf \Omega_i \in \mathfrak{so}(3)$ and $\mf v_i \in \mathbb R^3$ to refer to the rotational and translational components of $X_i$.
We have the following formulae, the algebraic details of which can be found in the supplementary material:
\begin{align}
    X_i^R \ell(g) &= (\mf \Omega_i \mf a + \mf v_i)^\top
    (\mf a - \mf p - \mf R x)\\
    F(X_i^R, X_j^R)_g &= (\mf \Omega_j \mf a + \mf v_j)^\top
    (\mf \Omega_i \mf a + \mf v_i ),
    \label{eq:second-deriv}
\end{align}
which allow us to compute the gradient and {\bf FIM} of $\ell$ in an arbitrary basis. 

If we have two landmarks $\mf a_1,\mf a_2 \in \mathbb R^3$ (such that $\mf a_1 \neq \mf a_2$) which each generate a measurement independently, then the combined likelihood function $f(g)$ is given by a product pdf. While the measurement of each landmark is invariant to transformations of the form $\left(\mf Q, (\mf I - \mf Q)\mf a_i\right)$, the transformations which leave \textit{both} measurements statistically unchanged are those that satisfy $(\mf I - \mf Q)\mf a_1 = (\mf I - \mf Q)\mf a_2$, that is, $ \mf Q(\mf a_1 - \mf a_2) = \mf a_1 - \mf a_2$. Such transformations correspond to rotations of $\mathbb R^3$ about the axis passing through $\mf a_1$ and $\mf a_2$, which is a one-dimensional subgroup of $SE(3)$.
Thus, for the two-landmark problem, we have $n_H=1$ and $n_{\Theta}= 5$. 

In Fig. \ref{fig:multi_Start}, we visualize the MLE of the robot's pose using two landmarks, each of which generates $10,000$ independent measurements. The generalized Fisher scoring algorithm (\ref{eq:fisher-scoring}) is used to compute the MLE across 4 different simulations. For each simulation, a semi-transparent red frame is used to show the initial estimate $\hat g_0$, and a line is used to depict the trajectory of estimates $(\hat g_1,\hat g_2,\ldots,\hat g_{k})$ which converges to the MLE.
While different initializations of the algorithm result in different estimates in $G$, these correspond to the same estimate on $H\backslash G$, i.e., they differ by an element of $H$. Figure \ref{fig:fisher-scoring} depicts the log-likelihood function $\ell(\mf x | \hat g_k)$ (whose normalization constant is neglected) of the estimates computed using the generalized Fisher scoring algorithm. The algorithm maximizes $\ell(\mf x|\,\cdot\,)$ (i.e., computes the MLE) in a few iterations despite not needing any tuning parameters aside from the initial estimate, $\hat g_0$ (which was chosen as the identity matrix). In contrast, a first-order gradient ascent approach for computing the MLE requires a {step-size} that tends to $0$, whose initial value and decay rate need to be tuned on a heuristic basis.
Finally, Figure \ref{fig:check_FIM} shows the variance of the MLE as well as the corresponding lower bound (obtained by taking the {trace} of inequality (\ref{eq:crb-GH-unbiased})) as a function of the number of measurements. The plot uses a $\log$ scale for either axis.
To compute the expectations, we averaged across $100,000$ Monte Carlo simulations.
The variance of the MLE on $H\backslash G$ (which was defined in Section \ref{sec:variance}) approaches the CRB.
On the other hand, the usual group-theoretic variance on $G$, which is defined as $\mathbb E\left[\|{\bm \upeta}(g)\|^2\right]$, does not approach the CRB, since it does not account for the invariance of the statistical model to the transformations in $H$.
\begin{figure}[h]
    \centering
    \includegraphics[width=0.5\textwidth, trim=4.0cm 5.7cm 2.9cm 4.35cm, clip]{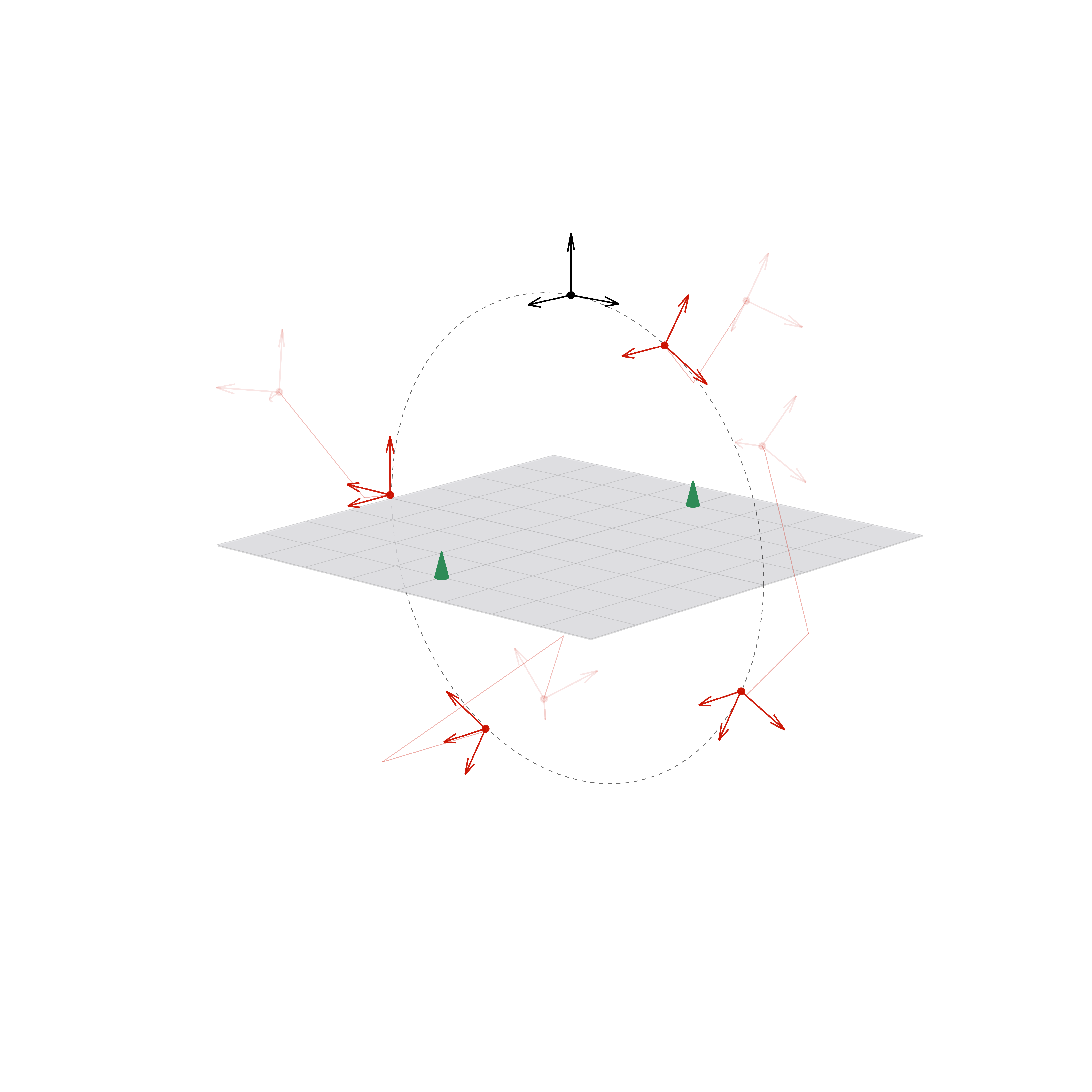}
    \caption{(\textit{Pose-Estimation using Two Landmarks}) The {black frame} represents the true pose $g \in SE(3)$ of the robot while the {green cones} represent landmarks. The circle depicts the set of poses from which the observed displacements of the landmarks are identical. The red frames represent the estimates of $g$ computed using the Fisher scoring algorithm for different initializations of the algorithm.}
    \label{fig:multi_Start}
\end{figure}
\begin{figure}[h]
    \centering
    \includegraphics[width=0.49\textwidth, trim=0.45cm 0.42cm 0.23cm 0.26cm, clip]{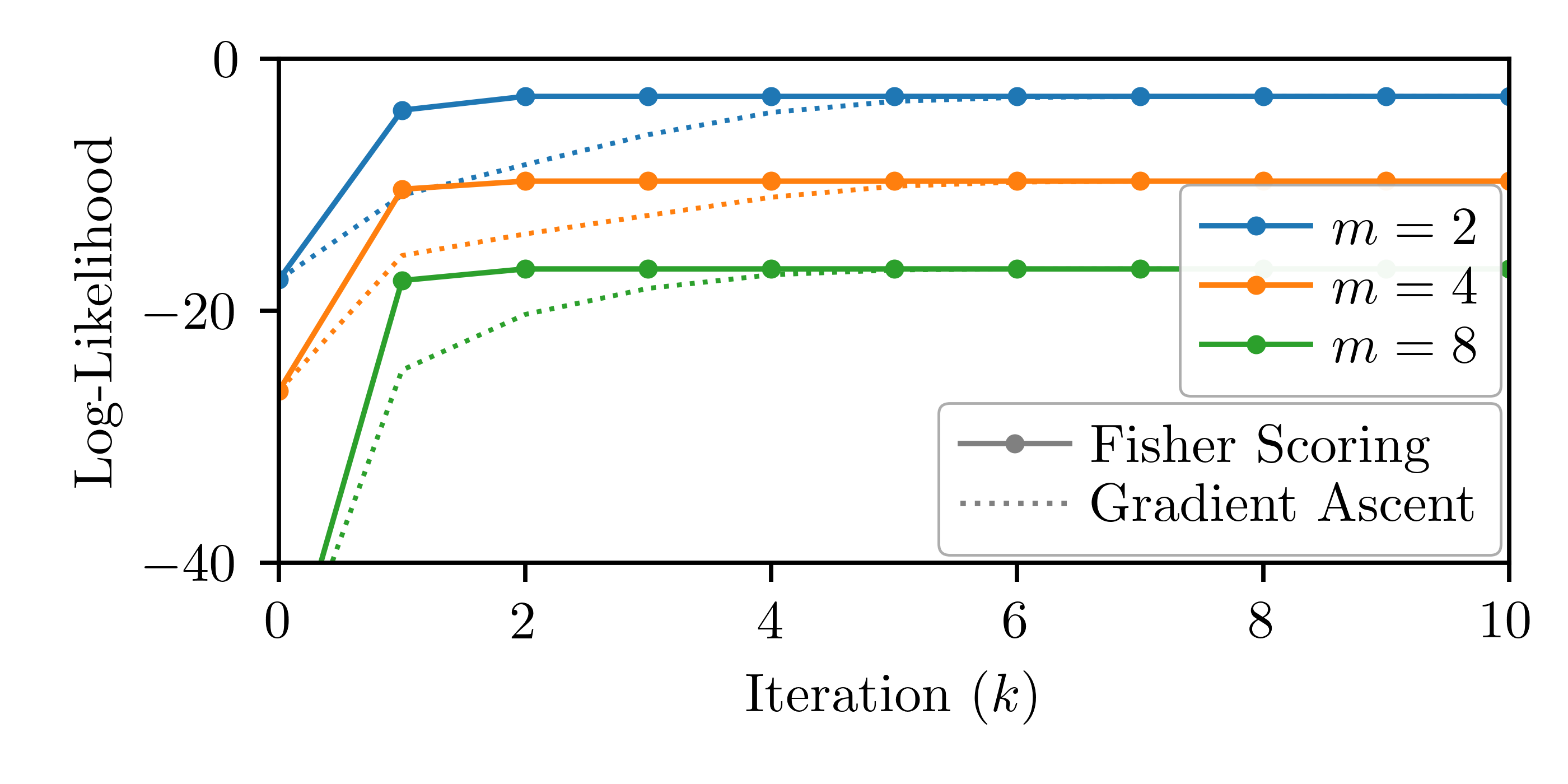}
    \caption{The generalized Fisher scoring algorithm (\ref{eq:fisher-scoring}) is compared with a first-order gradient ascent algorithm for computing the Maximum Likelihood Estimator (MLE) by maximizing the log-likelihood function. The simulation is repeated thrice for different values of $m$, where $m$ denotes the number of independent measurements generated by each landmark.}
    \label{fig:fisher-scoring}
\end{figure}

\begin{figure}[h]
    \centering
    \includegraphics[width=0.49\textwidth, trim=0.45cm 0.42cm 0.23cm 0.2cm, clip]{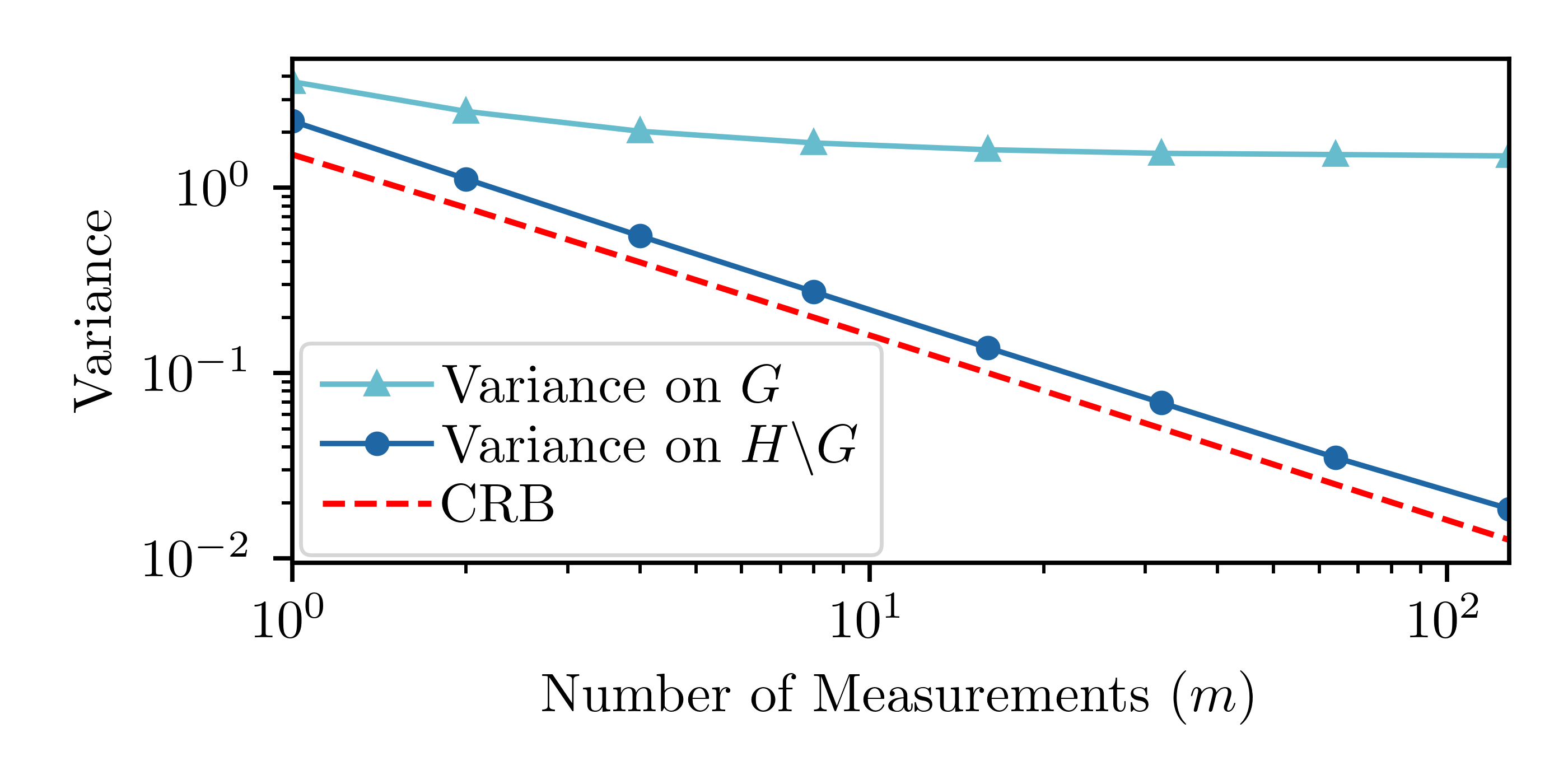}
    \caption{The variance of the MLE is plotted against the number of measurements using a $\log$ scale for either axis. The variance of the MLE on $G$, defined as $\mathbb E\left[\|{\bm \upeta}(g)\|^2\right]$, does not approach the CRB. On the other hand, the quantity $\mathbb E\left[\|\bar {\bm \upeta}(g)\|^2\right]$, which is the variance on $H\backslash G$ as defined in Section \ref{sec:variance}, approaches the CRB as $m \rightarrow \infty$.}
    \label{fig:check_FIM}
\end{figure}

\begin{remark}[Relationship to the Riemannian CRB]
It is possible to repeat the above analysis using the Riemannian definitions used in \cite{smith2005covariance} and \cite{boumal2013intrinsic}. However, the computations are significantly more involved. While the $\exp(\,\cdot\,)$ map of $SE(3)$ is readily implemented using the closed-form formulae given in \cite{chirikjian2011stochastic}, it is different from the Riemannian exponential map (regardless of the choice of Riemannian metric on $SE(3)$ \cite{zefran1999metrics}), which requires solving an ordinary differential equation at each point. See Appendix \ref{app:relationship} for more details.
\end{remark}


\subsubsection*{\textbf{Example 2.} Sensor Network Localization using Distances}
As a second example, consider the problem of localizing (i.e., estimating the spatial configuration of) a sensor network using the measured distances between the agents. The sensor network is represented by the graph $(\mathcal V, \mathcal E)$ where the vertices $\mathcal V =\lbrace 1, 2, \ldots, |\mathcal V|\rbrace$ represent the agents (with $|\,\cdot\,|$ denoting the cardinality of a set), and the edges $\mathcal E \subseteq \mathcal V \times \mathcal V$ represent the pairs of agents who can measure their distance from each other.
The position of the $i^{th}$ agent is $\mf p_i \in \mathbb R^d$, where we typically set $d=2$ or $3$ depending upon the application \cite{zelazo2015decentralized,cano2021improving}.
Corresponding to the $i^{th}$ agent, we define the transformation $g_i\in SE(d)$, which we will also write as $(\mf R_i, \mf p_i)$. Here, the rotation $\mf R_i$ can be viewed as an auxiliary variable; it helps us describe the underlying homogeneous space as a coset space.
The collective configuration of the sensor network is given by $g=(g_1, g_2, \ldots, g_{|\mathcal V|}) \in SE(d)^{|\mathcal V|}$.

Let $\varphi: SE(d) \rightarrow \mathbb R_{\geq 0}$ be the mapping $(\mathbf R, \mathbf p) \mapsto \frac{1}{2}\lVert \mf p \rVert^2$, with $\lVert \, \cdot \, \rVert$ denoting the vector 2-norm. 
Observe that
$\varphi (g_i^{-1}g_j) = \frac{1}{2}\lVert \mf p_i - \mf p_j \rVert^2$ is a measure of the distance between agents $i$ and $j$.
Letting $x_{ij}$ represent the measurement obtained along the edge $(i,j)\in \mathcal E$, we consider the following statistical model for the measurements:
\begin{align}
    f(x|g) \propto \exp\bigg(-\sum_{(i,j) \in \mathcal E}\frac{1}{2\sigma_{ij}^2} \Big(x_{ij} - \varphi (g_i^{-1} g_j)\Big)^2\bigg),
\end{align}
which assumes (for the sake of simplicity) that the measurement noise at each edge is Gaussian and independent of the measurements at the other edges.
It is readily verified that $\varphi (g_i^{-1}g_j) = \varphi ((hg_i)^{-1}(hg_j))$
for all $i,j \in \mathcal V$.
This observation leads us to consider the following subgroup of $SE(d)^{|\mathcal V|}$, representing the group of symmetries w.r.t. which the measurement model is invariant:
\begin{align}
    \tilde H \coloneqq \lbrace (h,h, \ldots,h) \, |\, h \in SE(d) \rbrace \; \subseteq \; SE(d)^{|\mathcal V|}.
    \label{eq:subgroup}
\end{align}
The subgroup $\tilde H$ corresponds to rigid translations and rotations of the entire sensor network, which leaves the pairwise distances between the agents invariant.
However, there are additional symmetries which must be included, each corresponding to a mapping of the form $(\mf R_i, \mf p_i) \mapsto (\mf Q \mf R_i, \mf p_i)$ which does not change agent $i$'s position in $\mathbb R^d$. Let $H$ denote the subgroup of $SE(d)^{|\mathcal V|}$ which includes these symmetries in addition to those in $\tilde H$.
Thus, we have a parameter estimation problem on the homogeneous space $\Theta \cong H \backslash SE(d)^{|\mathcal V|}$.

Recall that the dimensions of $SE(d)$ and $SO(d)$ are $\binom{d+1}{2}$ and $\binom{d}{2}$, respectively, where $\binom{n}{k}$ represents the binomial coefficient with indices $n$ and $k$.
Therefore, the dimension of $G$ is $n_G = \binom{d+1}{2}|\mathcal V|$, while the dimension of $H$ is $n_H = \binom{d}{2}|\mathcal V| + \binom{d+1}{2}$ corresponding to the rotations of the individual agents and the subgroup $\tilde H$ described in (\ref{eq:subgroup}), respectively. The dimension of $\Theta$ (which is also the dimension of $\mf m$) is computed as $n_{\Theta} = n_G - n_H$, which is found to be $n_{\Theta}=d|\mathcal V| - \binom{d+1}{2}$.

Specializing to the case of $d=2$, we see that $\mf m$ is a $2|\mathcal V| - 3$ dimensional vector space. One way to choose a basis for $\mf m$ is as follows. We assume without loss of generality that $\mf p_1$ and $\mf p_2$ lie on the $\textrm{y}$-axis of our coordinate system.
Consider the $2|\mathcal V|$ vectors in $\mathfrak g$ that represent the $2$-dimensional translation of each of the $|\mathcal V|$ agents. We discard the first $3$ vectors in this sequence, resulting in a total of $2|\mathcal V| - 3$ basis vectors as desired. This choice of basis represents all the perturbations of the sensor network in which agent $1$'s coordinates are held fixed, while only the $\textrm{y}$-coordinate of agent $2$ is allowed to vary.
On the other hand, we can consider what would happen if the algorithm were to update along the directions in $\mathfrak h$. When the estimate is updated using a vector in $\mathfrak h$, the entire sensor network is moved in a rigid manner, as depicted by the green configuration in Fig. \ref{fig:distances}.
This suggests that additional measurements and/or regularization techniques can be used to update the estimate along $\mathfrak h$, for e.g., a sparsity-promoting regularization term is used in \cite{khan2023recovery} to resolve the ambiguity along $\mathfrak h$.

\begin{figure}[h]
    \centering
    \includegraphics[width=0.49\textwidth, trim=2.3cm 4.7cm 1.9cm 4.9cm, clip]{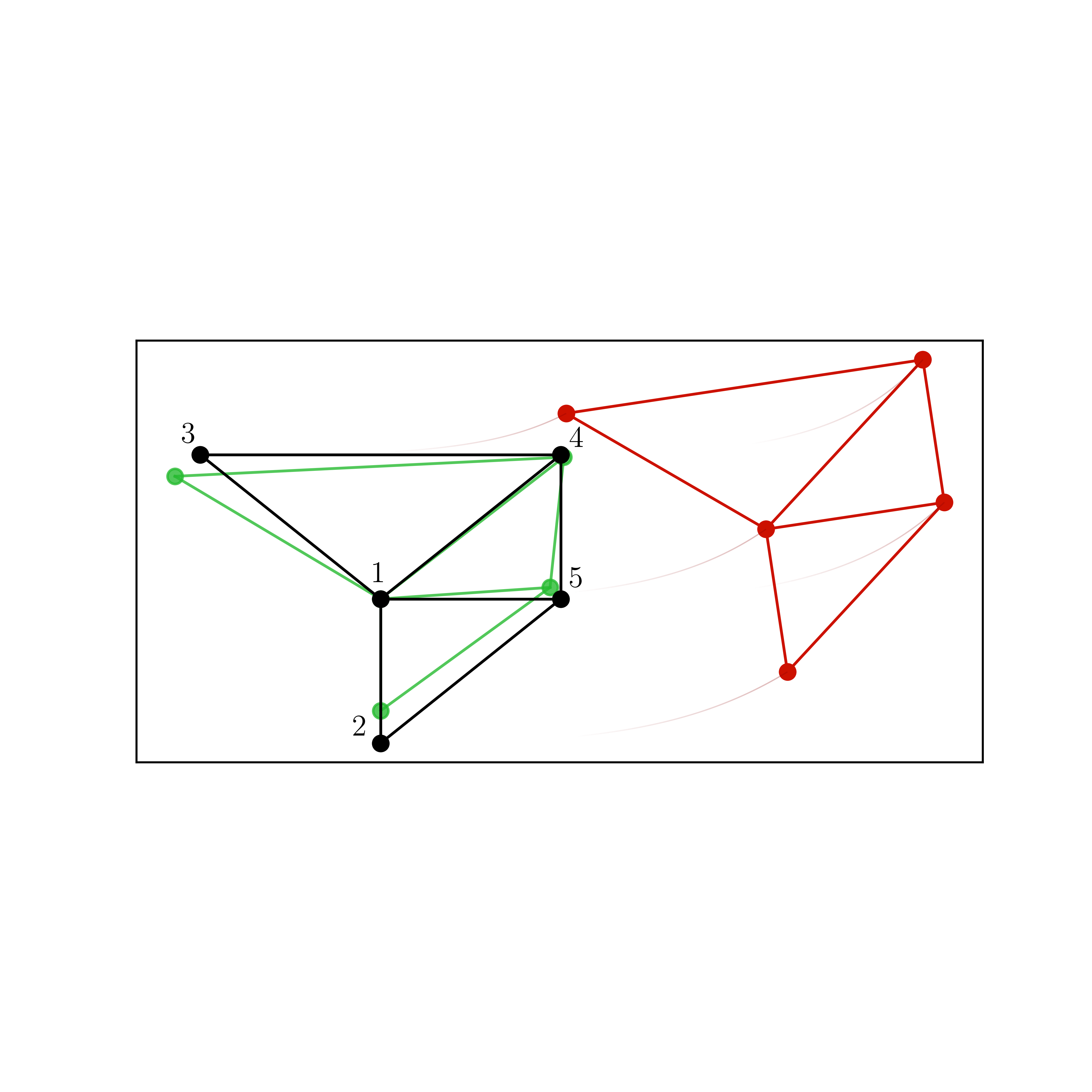}
    \caption{The agents of the sensor network are shown as black dots while the solid lines represent the edges between them. The black and green configurations are related by a transformation of the form $\exp(Z)$ with $Z\in \mf m$.
    The black and red configurations are related by an element of $H$, which does not change any of the edge lengths.}
    \label{fig:distances}
\end{figure}

It can be shown that the {\bf FIM} is a submatrix of the so-called \textit{Symmetric Rigidity Matrix} $\mf S_g \in \mathbb R^{2|\mathcal V|\times 2|\mathcal V|}$ \cite{zelazo2015decentralized,le2018localizability}, whose $(i,j)^{th}$ block is defined as
\begin{align*}
    {\mf S}_{ij,g} =
    \begin{cases}
        \sum_{k\in \mathcal N_i}\frac{1}{\sigma_{ik}^2}
        \left(\mf p_i  - \mf p_k \right)\left(\mf p_i  - \mf p_k \right)^\top,\ &i=j\\
        -\frac{1}{\sigma_{ij}^2}
        \left(\mf p_i  - \mf p_j \right)\left(\mf p_i  - \mf p_j \right)^\top,\ &(i,j)\in \mathcal E\\
        \mf 0 &\text{ otherwise}.
    \end{cases}
\end{align*}
By discarding the first three rows and columns of $\mf S_{g}$, we obtain the {\bf FIM}, $\overline{\mf F}_g$. It is well-known that the rank of $\mf S_g$ is at most $2|\mathcal V| - 3$ (for the case of $d=2$) and that its smallest non-zero eigenvalue is closely related to the robustness of sensor network localization algorithms \cite[Thm. 4]{khan2023recovery}. Indeed, our analysis above has shown that the accuracy of sensor network localization depends on the smallest non-zero eigenvalue of $\overline{\mf F}_g$, which is precisely the spectral radius of the matrix $\overline{\mf F}_g^{-1}$ that appears in the CRB.


\section{Conclusion}\label{sec:conclusion}
In this article, we developed the theory of parameter estimation on homogeneous spaces by introducing group-theoretic definitions of the estimation error, bias, and covariance. By characterizing the properties of the Fisher Information Metric (FIM), we proceeded to derive exact and approximate Cram\'er-Rao Bounds (CRBs) on Lie groups and homogeneous spaces. The relationship of these CRBs with the intrinsic (Riemannian) versions was also established.
Illustrative applications were used to demonstrate that the theory can be used to tackle parameter estimation problems where the parameters of interest are \textit{invariant} with respect to the action of a given Lie group.

Future work can consider specializing the Bayesian CRB (also called the Van Trees CRB) to homogeneous spaces. Another promising line of research is to theoretically investigate the convergence properties of the \textit{generalized Fisher scoring} algorithm presented in Section \ref{sec:applications}.

\appendices
\section{Riemannian vs. Group-Theoretic CRB}\label{app:relationship}
In this section, we will make precise the relationship between the intrinsic (Riemannian) CRB obtained in \cite{smith2005covariance} and the group-theoretic CRB developed in this article. 
Suppose we have an estimator $\hat \theta : \mathcal X \rightarrow G/H$ which is sufficiently accurate (i.e., concentrated around the true parameter) and $G/H$ is endowed with a Riemannian metric. Then, the \textit{intrinsic error} of the estimator can be defined as 
    $\bar \eta^{\textrm{int.}}(\theta) \coloneq \log_{\theta}\hat \theta \;\in\; T_{\theta}\, G/H,$
where $\log_{\theta}$ is the \textit{Riemannian log map} at $\theta$, with $\log_{\theta}\hat \theta$ representing the initial velocity of the geodesic from $\theta$ to $\hat \theta$. If $G/H$ admits a $G$-invariant metric, then the $G$-invariant metric is a natural and commonly used choice of Riemannian metric for $G/H$ \cite{gallier2020differential,andreas_arvanitogeorgos2003}. 
Sufficient conditions under which the intrinsic CRB in\cite{smith2005covariance} (which is stated in terms of $\bar \eta^{\textrm{int.}}(\theta)$) coincides with the group-theoretic CRB of Theorem \ref{thm:crb-group-homogeneous} are recorded in the following lemma.

\begin{lemma}
Let $\theta = gH$ be a parameter and $g\in G$ any representative of $gH$. Suppose that 
\begin{enumerate}
    \item $G$ admits a bi-invariant metric, and 
    \item $G/H$ is a naturally reductive space,\footnote{The definition of a \textit{naturally reductive} space can be found in differential geometry texts \cite{cheeger1975comparison,gallier2020differential,andreas_arvanitogeorgos2003}. In particular, \textit{Riemannian symmetric spaces} are naturally reductive \cite[Sec. 23.8]{gallier2020differential}, \cite[Ch. 6]{andreas_arvanitogeorgos2003}.}
\end{enumerate} 
    then
    $
        \bar{\bm \upeta}^{\textrm{int.}}(gH) = \bar {\bm \upeta}(g),
    $
    where $\bar{\bm \upeta}^{\textrm{int.}}(gH)$ represents the intrinsic error expressed in the basis for $T_{gH}(G/H)$ given in (\ref{eq:LIVF_basis_on_GH}).
    \label{lem:conditions}
\end{lemma}
\begin{proof}
Condition 1 of the lemma is also required for the Lie group CRBs in \cite{bonnabel2015intrinsic,labsir2023barankin} to coincide with the intrinsic CRB. 
When the $G$ admits a bi-invariant metric, the map
\begin{align*}
    \log_g :\ 
    G &\rightarrow \;T_g G\\
    \tilde g \,&\mapsto g\log(g^{-1}\tilde g)
\end{align*}
is precisely the Riemannian log map \cite[Corr. 3.19]{cheeger1975comparison}. Lemma 7.5 of \cite{milnor1976curvatures} completely characterizes the Lie groups satisfying condition 1, which include all the compact Lie groups like $SO(3)$. The relevance of condition 2 is explained below.

Let $\hat g \coloneq g \exp \hat Y$ be the horizontal lift of $\hat \theta$ at $g$. Its group-theoretic error is $\eta(g)= \log(g^{-1} \hat g)=\hat Y \in \mf m$. Due to bi-invariance of the metric, $g \exp t\hat Y$ (with $t\in[0, 1]$) is the length-minimizing geodesic from $g$ to $\hat g$ \cite[Corr. 3.19]{cheeger1975comparison}. Moreover, since $\hat Y\in \mf m$ and $G/H$ is naturally reductive, the curve $g\exp(t \hat Y)H$ is a geodesic connecting $\theta$ and $\hat \theta$ on $G/H$ \cite[Prop. 23.28]{gallier2020differential}. The initial velocity of this geodesic is $\bar \eta^{\textrm{int.}}(gH)$ by definition, and is given by
\begin{align*}
     \frac{d}{dt} g\exp(t \hat Y)H \Big|_{t=0} = \frac{d}{dt} \pi\big(g\exp(t \hat Y)\big) \Big|_{t=0} = d \pi_g \,(\hat Y^L_g).
\end{align*}
When, expressed in the basis given in (\ref{eq:LIVF_basis_on_GH}), we have $\bar {\bm\upeta}^{\textrm{int.}}(gH) = \bm \Pi\, \hat Y^\vee =  \bm \Pi\, \bm \upeta(g)  =\bar{\bm \upeta}(g)$, which completes the proof.
\end{proof}

Finally, we remark that the $n$-dimensional sphere $S^n$ and the Grassmannian manifolds satisfy the conditions of Lemma \ref{lem:conditions} \cite{gallier2020differential} but the examples considered in Section \ref{sec:applications} do not, since $SE(3)$ does not admit a bi-invariant metric.
\section{Proof of Theorem \ref{thm:crb-group}}\label{app:proof-crb-group}
\noindent
First, we recall the following identity.
\begin{lemma}[Derivative of $\log$\cite{varadarajanlie,hall2013lie}]
    For any $X, Y\in \mathfrak g$, we have
    \begin{align}
        \frac{d}{dt} \log\big(\exp(X)\exp(tY)\big)\bigg|_{t=0} &= \Psi_X(Y)\\
        \frac{d}{dt} \log\big(\exp(tY)\exp(X)\big)\bigg|_{t=0} &=
            \Psi_{-X}(Y)
    \end{align}
    where 
    \begin{align}
        \Psi_X
        &= \frac{\ad_{X}}{{\mf I - \exp(-\ad_{X})}} = \mf I + \frac{1}{2}\ad_X + \sum_{n=1}^\infty \frac{\beta _{2n}}{(2n)!}(\ad_{X})^{2n}. \label{eq:varadarajan-formula}
    \end{align}
    Here, $\beta_n$ denotes the $n^{th}$ Bernoulli number.
    \label{lem:derivative-log}
\end{lemma}
The first equality in (\ref{eq:varadarajan-formula}) was shown in \cite[Sec. 5.5]{hall2013lie}, whereas the second equality was shown in \cite[p. 117]{varadarajanlie} as well as in \cite[Sec. 8.1]{muger2019notes}. We let $\mf \Psi_{(\,\cdot\,)}$ denote the matrix representation of $\Psi_{(\,\cdot\,)}$ with respect to the basis $\lbrace E_i \rbrace_{i=1}^{n_G}$.
\begin{remark}[Bernoulli Numbers]\label{rem:bernoulli}
    The first three Bernoulli numbers are $\beta_0=1$, $\beta_1=\pm {1}/{2}$, and $\beta_2=1/6$, while $\beta_n=0$ for all odd $n>2$. The sign of $\beta_1$ depends on the convention used, for e.g., \cite{muger2019notes} and \cite{labsir2024intrinsic} use $-{1}/{2}$ while \cite{arakawa2014bernoulli} uses $+{1}/{2}$. 
    Note that (\ref{eq:varadarajan-formula}) does not depend on the convention used to define $\beta_1$.
\end{remark}
%
%
%
%
%
%
\begin{proof}[Proof of Theorem \ref{thm:crb-group}]
From the definition of the (left) bias (\ref{eq:bias}), we have
\begin{align}
    \int_{\mathcal X} \left(\eta(g) - B(g) \right) f(g) \, d\mu = 0
\end{align}
where $0$ represents the zero vector of $\mathfrak g$.
Let $X^L = x^i E_i^L$ be an arbitrary LIVF on $G$, such that $x^i$ is the $i^{th}$ component of $\mf x \in \mathbb R^{n_G}$ and $\mf x = X^\vee$.
Operating on both sides by $X^L(\,\cdot\,)$, we get
\begin{align}
    -\int_{\mathcal X} \Big(X^L &\eta\,(g) - X^L B(g)\Big) f(g) \, d\mu 
    \nonumber\\
    &= \int_{\mathcal X} \big(\eta(g) - B(g)\big)\,X^L \ell(g)\,f(g) \, d\mu
    \label{eq:derivative-bias}
\end{align}
where we used the identity, $X^L f(g) = f(g)\,X^L \ell(g)$. In (\ref{eq:derivative-bias}), the quantity $X^L \eta$ represents the Lie derivative of $\eta$ along $X^L$:
\begin{align}
    X^L \eta\,(g) &= \frac{d}{dt} \log\Big(\big(\,g \exp (t X)\big)^{-1} \hat g\Big)\bigg|_{t=0}\\
    &= \frac{d}{dt} \log\Big(\, \exp (-t X) \exp\left(\eta(g)\right)\Big)\bigg|_{t=0}\\
    &= \Psi_{-\eta(g)} (-X).
\end{align}
Returning to (\ref{eq:derivative-bias}), we can write everything in terms of vectors by using $(\,\cdot\,)^\vee$ on both sides:
\begin{align}
    \mathbb E \big[\boldsymbol\Psi_{-\eta(g)}\big] \mf x &+ \mf J_{\mf b}(g) \mf x
    = \int_{\mathcal X} \left(\bm \upeta(g) - \mf b(g)\right) X^L \ell(g) f(g) d\mu
    \label{eq:crb-proof-vector-form}
\end{align}
Choosing an arbitary vector $\mf y \in \mathbb R^{n_G}$, we operate on both sides of (\ref{eq:crb-proof-vector-form}) with $\mf y^\top(\,\cdot\,)$. The term on the right-hand side can be squared, and the Cauchy-Schwarz inequality (for the $L^2$ function space) can be used to obtain
\begin{align}
   \bigg(\int_{\mathcal X} &\mf y ^\top \big(\bm \upeta(g) - \mf b(g)\big)\sqrt{f(g)}\; X^L \ell(g) \sqrt{f(g)} d\mu\bigg)^2 \nonumber\\
   \leq &\left(\int_{\mathcal X} \mf y ^\top \left(\bm \upeta(g) - \mf b(g)\right)\,\left(\bm \upeta(g) - \mf b(g)\right)^\top \mf y\, f(g) d\mu\right)\nonumber\\
   &\quad\cdot \left(
 \int_{\mathcal X} x^i E_i^L \ell(g)\; x^j E_j^L \ell(g) \,f(g) d\mu\right)
\end{align}
After rearranging the terms and substituting into (\ref{eq:crb-proof-vector-form}), we get
\begin{align}
    \big(\mf y^\top \bm \Phi(g) \mf x\big)^2
    \;\leq\; \mf y^\top \mf \Sigma(g) \mf y \cdot \mf x^\top \mf F_g^L \mf x,
\end{align}
where $\bm \Phi(g)$ is as defined in the Theorem.
Finally, we set $\mf x = {\mf F^L_g}^\dagger \,\bm \Phi(g)^\top \mf y$, where $(\,\cdot\,)^\dagger_g$ denotes the Moore-Penrose pseudoinverse of a matrix, giving us
\begin{align}
    \mf y^\top \bm \Phi(g){\mf F^L_g}^\dagger \,\bm \Phi(g)^\top \mf y
    \;\leq\; \mf y^\top \mf \Sigma(g) \mf y 
\end{align}
for all $\mf y \in \mathbb R^n$, which proves the inequality.

When the right-invariant error $\eta'(g)$ is used, the proof is similar except
we use $X^R$ instead of $X^L$ for the differentiation.
In this case, the term $X^R \eta' (g)$ incurs a negative sign:
$X^R \eta'\,(g) = \Psi_{\eta'(g)} (-X)$.
\end{proof}

\bibliographystyle{IEEEtran}
\bibliography{sections/references}

\end{document}


\title{Supplementary Material for ``\textit{Parameter Estimation on Homogeneous Spaces}"}

\author{Shiraz Khan, Gregory S. Chirikjian,~\IEEEmembership{Fellow,~IEEE}%
\thanks{Corresponding email: \tt{shiraz@udel.edu}.}%
}

\maketitle

In this document, we provide supplementary details about the applications of the Fisher-Rao theory on homogeneous spaces. We begin by deriving the Hessian form of the Fisher Information Matrix ({\bf FIM}). Recall that $\bar f(gH) \coloneq f(\,\cdot\,|\,gH)$ is a pdf on $G/H$. We define $f\coloneq \bar f \circ \pi$ as the pullback of $\bar f$ to $G$, and $\ell(g) \coloneq \log\left(f(g)\right)$ is the log-likelihood function. 
\section{Hessian Form of the {\bf FIM}}
Since $f(g)$ is a pdf, we have
\begin{align}
    \int_{\mathcal X}f(g)d\mu &= 1
\end{align}
Differentiating this expression w.r.t. $E_i^L$, we have
\begin{align}
\left[E_i^L\int_{\mathcal X}f(\,\cdot\,)d\mu \right](g)
=\int_{\mathcal X}&[E_i^L f](g)d\mu = 0.
\label{eq:f_deriv}
\end{align}
Using the fact that $$[E_i^L \ell](g) = \frac{[E_i^L f](g)}{f(g)},$$
we can rewrite (\ref{eq:f_deriv}) as
\begin{align}
    \int_{\mathcal X}[E_i^L \ell](g)\, f(g)\,d\mu &= 0.
\end{align}
Differentiating again w.r.t. $E_j^L$ and using the product rule, we get
\begin{align}
    \int_{\mathcal X}[E_j^L E_i^L \ell](g)\, f(g)\,d\mu +\int_{\mathcal X} [E_i^L \ell](g)\, &[E_j^L \ell](g)\, f(g)\,d\mu \nonumber \\&= 0
    \label{eq:fsecond-deriv}
\end{align}
where we have used $[\,\cdot\,]$ to clarify which function each derivative operator acts on.
Since the second term on the left is precisely $\mf F_{ij,g}^L$ (which is defined in the main paper), we obtain the Hessian form of the {\bf FIM}:
\begin{align}\mf F_{ij,g}^L = - \mathbb E \left[E_j^L E_i^L \ell(g)\right]=- \mathbb E \left[E_i^L E_j^L \ell(g)\right].\end{align}
While the derivatives $E_i^L$ and $E_j^L$ do not commute in general, the last inequality holds due to (\ref{eq:fsecond-deriv}).

\begin{remark}
    Note that the {\bf FIM} can be computed purely by using ordinary second derivatives.
    The diagonal components of the {\bf FIM} can be computed as
\begin{align*}\mf F_{ii,g}^L = - \mathbb E \left[E_i^L E_i^L \ell(g)\right] = 
- \mathbb E\left[
\frac{d^2}{dt^2} \ell\big(g \exp(t E_i)\big)\bigg|_{t=0}
\right]
\end{align*}
whereas the off-diagonal components may be obtained by exploiting the bilinearity of $F$:
\begin{align*}\mf F_{ij,g}^L &=F(E_i, E_j)\\&= \frac{1}{4}\big(F(E_i + E_j, E_i + E_j) - F(E_i - E_j, E_i - E_j)\big). \end{align*}
\end{remark}

\section{Pose Estimation using Landmarks}
\subsection{Problem Formulation}
Let $g\in SE(3)$ describe the \textit{pose} (i.e., position and orientation) of a mobile robot. We write $g = (\mf R, \mf p)$, such that $\mf R \in SO(3)$ and $\mf p \in \mathbb R^3$ are understood to refer to the rotational and translational components of $g$, respectively. To estimate its own pose, the robot uses a \textit{landmark}, whose position $\mf a \in \mathbb R^3$ is known in a global frame of reference. The location of the landmark is measured in the robot's local frame of reference as the $\mf x \in \mathcal X = \mathbb R^3$ which is subject to standard Gaussian noise, yielding the likelihood function
\begin{align}
    f(\mf x\,|g) \propto \exp\left(-\frac{1}{2}\left\lVert \mf x - \mf R^\top (\mf a  - \mf p) \right\rVert^2\right)
\end{align}
with $\lVert \, \cdot \, \rVert$ denoting the Euclidean norm \cite{chirikjian2014gaussian}. 
The quantity $\mf R^\top (\mf a  - \mf p)$ can also be written as follows:
\begin{align}
g^{-1} \begin{bmatrix}
    \mf a\\
    1
\end{bmatrix} = \begin{bmatrix}
    \mf R^\top & -\mf R^\top \mf p\\
    \mf 0 & 1
\end{bmatrix} \begin{bmatrix}
    \mf a\\
    1
\end{bmatrix} = \begin{bmatrix}
    \mf R^\top \mf a - \mf R^\top \mf p\\
    1
\end{bmatrix},
\label{eq:landmark-in-robot-frame}
\end{align}
which is useful for computing the gradient as well as for capturing the invariances/symmetries of the log-likelihood function. In particular, note that $f(\mf x|g) = f(\mf x|hg)$ for all $h\in H$, where $H\subseteq SE(3)$ is defined as follows:
\begin{align}H\coloneq\left\lbrace \big(\mf Q, (\mf I - \mf Q)\mf a\big)\,\big|\,\mf Q \in SO(3)\right\rbrace.\label{eq:symmetry-group}\end{align}
To verify this, we observe that
\begin{align}
    (h g)^{-1} \begin{bmatrix}
    \mf a\\
    1
\end{bmatrix} &= 
g^{-1}
\begin{bmatrix}
    \mf Q^\top & -\mf Q^\top (\mf I - \mf Q)\mf a\\
    \mf 0 & 1
\end{bmatrix}
\begin{bmatrix}
    \mf a\\
    1
\end{bmatrix} \\
&= g^{-1}
\begin{bmatrix}
    \mf Q^\top\mf a \mathrel-\mf Q^\top \mf a \mathrel+ \mf Q^\top\mf Q \mf a\\
    \ 1
\end{bmatrix}\\
&= 
g^{-1} \begin{bmatrix}
    \mf a\\
    1
\end{bmatrix}.
\end{align}
Thus, we can view the pose estimation problem w.r.t. a single landmark as a parameter estimation problem on $H \backslash SE(3)$, where $H \cong SO(3)$ is the subgroup of $SE(3)$ that leaves the landmark $\mf a\in\mathbb R^3$ fixed. 
Note that $H$ is three-dimensional, so $n_H=3$ and $n_{\Theta}= n_G - n_H = 3$.

\subsection{Computing the Gradient and {\bf FIM}}
Consider an arbitrary matrix in $\mathfrak{se}(3)$ of the form
\begin{align}
    X = \begin{bmatrix}
        \mf \Omega & \mf v\\
        \mf 0 & 0
    \end{bmatrix} \in \mathfrak{se}(3).
\end{align}
To evaluate the gradient of $\ell$, we need Lie derivatives of the following form:
\begin{align}
X^R \ell(g) &= \frac{d}{dt} \ell(\exp(tX)g)\Big|_{t=0} \\
&=
-\frac{d}{dt} \frac{1}{2} \left\lVert \begin{bmatrix}\mf x\\1\end{bmatrix} - g^{-1}(\mf I-tX)\begin{bmatrix}
    \mf a\\
    1
\end{bmatrix} \right\rVert^2 \bigg|_{t=0}
\label{eq:XR-ell}
\end{align}
where we used (\ref{eq:landmark-in-robot-frame}) the fact that $\exp(tX)^{-1} = \exp(-tX) = \mf I - tX + O(t^2)$.\footnote{Note that the last component (i.e., `$1$') of the vectors in (\ref{eq:XR-ell}) disappears when we take the derivative. Hence, its inclusion does not change the derivative, but it lets us write the calculation in terms of $g$.}
We have
\begin{align}
    (\mf I-tX)\begin{bmatrix}
        \mf a\\
        1
    \end{bmatrix} = 
    \begin{bmatrix}
        \mf a - t\,\mf \Omega \mf a - t\,\mf v\\
        1
    \end{bmatrix},
\end{align}
giving us, using the chain rule,
\begin{align}
    X^R \ell(g)
    &=
    -\frac{d}{dt} \frac{1}{2}\left\lVert \mf x - \mf R^\top (\mf a - \mf p) + t \,\mf R^\top(\mf \Omega \mf a + \mf v)\right\rVert^2 \bigg|_{t=0}\nonumber \\
    &= (\mf R^\top (\mf a - \mf p)-\mf x)^\top \,\mf R^\top(\mf \Omega \mf a + \mf v)\\
    &= (\mf \Omega \mf a + \mf v)^\top
    (\mf a - \mf p - \mf R \mf x).
\end{align}
This general formula allows us to compute the gradient of $\ell$ in any choice of basis, although we are especially interested in a basis for $\mf m$ that is orthonormal to that of $\mathfrak h$. 

Before we choose a basis, let us also compute the second derivatives needed to evaluate the {\bf FIM}. Let $X_i,X_j \in \mathfrak{se}(3)$ be two matrices with skew-symmetric and translational components $\mf \Omega_i, \mf \Omega_j \in \mathfrak{so}(3)$ and $\mf v_i, \mf v_j \in \mathbb R^3$, respectively. Noting that
\begin{align}
    \exp(tX_i) g \approx (\mf I + tX_i)g = \begin{bmatrix}
        \mf R + t\mf \Omega_i \mf R & \mf p + t(\mf \Omega_i \mf p + \mf v_i )\\
        \mf 0 & 1
    \end{bmatrix},
\end{align}
we have
\begin{align}
    X_i^R X_j^R \ell(g) &=  X_i^R\left[(\mf \Omega_j \mf a + \mf v_j)^\top
    (\mf a - \mf p - \mf R \mf x)\right](g)\nonumber\\
    &= \frac{d}{dt}(\mf \Omega_j \mf a + \mf v_j)^\top
    \Big(\mf a - \mf p - t(\mf \Omega_i \mf p + \mf v_i ) \nonumber\\&\hspace{3.8cm}- (\mf R + t\mf \Omega_i \mf R) \mf x\Big)\Bigr|_{t=0}\nonumber\\
    &= -(\mf \Omega_j \mf a + \mf v_j)^\top
    (\mf \Omega_i \mf p + \mf v_i + \mf \Omega_i \mf R \mf x)
\end{align}
so that
\begin{align}
    -\mathbb E\left[X_i^R X_j^R \ell(g)\right] &= (\mf \Omega_j \mf a + \mf v_j)^\top
    (\mf \Omega_i \mf p + \mf v_i + \mf \Omega_i \mf R\, \mathbb E[\,\mf x\,]) \nonumber\\
    &=  (\mf \Omega_j \mf a + \mf v_j)^\top
    (\mf \Omega_i \mf a + \mf v_i )\nonumber\\
    &= F(X_i^R, X_j^R)_g.
    \label{eq:second-deriv}
\end{align}
where we used $\mathbb E[\,\mf x\,] = \mf R^\top (\mf a - \mf p)$.
Note that while the derivatives $X_i^R$ and $X_j^R$ do not commute in general, the expected value in (\ref{eq:second-deriv}) is symmetric in $i$ and $j$ since it describes a symmetric covariant tensor (namely, the FIM). Another sanity check which verifies the formula in (\ref{eq:second-deriv}) is the positive definiteness of the FIM.

\subsection{Choosing Bases for $\mathfrak h$ and $\mf m$}
What is the Lie subalgebra, $\mathfrak h$? We recall from \cite{chirikjian2011stochastic} that the log map of $SE(3)$ is given by
\begin{align}
    \log(g)^\vee = \begin{bmatrix}
        \log (\mf R)^\vee\\
        \big[\,\mf J(\log \mf R)\,\big]^{-1} \mf p
    \end{bmatrix},
\end{align}
where we use $\vee$ to represent both the `vee' maps of $SE(3)$ and $SO(3)$ depending on the context, and (from \cite[p. 40]{chirikjian2011stochastic})
\begin{align}
    \big[\,\mf J(\mf \Omega)\,\big]^{-1} = \mf I - \frac{1}{2}\mf \Omega + \left(
        \frac{1}{\lVert \mf \Omega \rVert^2} - \frac{1 + \cos \lVert \mf \Omega \rVert}{2\lVert \mf \Omega \rVert \sin \lVert \mf \Omega \rVert}
    \right) \mf \Omega^2.
\end{align}
Thus, $\mathfrak h$ is characterized be elements of the form
\begin{align}
    \log(g)^\vee = \begin{bmatrix}
        \log (\mf Q)^\vee\\
        \big[\,\mf J(\log \mf Q)\,\big]^{-1} (\mf I - \mf Q) \mf a
    \end{bmatrix}.
\end{align}
By replacing $\log \mf Q$ with the usual basis vectors for $\mathfrak{so}(3)$, we obtain a basis for $\mathfrak h$ that has three elements as desired. The basis for $\mf m$ can be constructed using Gram Schmidt orthogonalization, which we do numerically for convenience.\footnote{The $Ad_H$ invariance of the inner product is also verified numerically; see the code on GitHub: \textrm{https://github.com/shirazkn/lie-groups/blob/main/parameter-estimation/landmarks.py}.} To initiate the Gram Schmidt process, we need to pick three linearly independent vectors which are not contained in $\mathfrak h$. For this, we can simply pick the three standard basis vectors of $SE(3)$ representing pure translations (i.e., no rotation) of space. Since pure translational motions do not have a fixed point, they certainly do not fix $\mf a$, and hence the corresponding tangent vectors do not lie in $\mathfrak h$.

\subsection{{\bf FIM} for Two Landmarks}
If we have two landmarks which each generate (independently) a measurement of the form considered in the previous subsections, then the likelihood function is
\begin{align}
    f(\mf x_1,\mf x_2\,|\,g) \propto \exp\left(-\sum_{i=1,2}\frac{1}{2}\left\lVert \mf x_i - \mf R^\top (\mf a_i  - \mf p) \right\rVert^2\right)
\end{align}
and it is clear how to compute its gradient and {\bf FIM}; they are the sums of the gradients and {\bf FIM}s of the individual landmarks, respectively. What is not obvious is how to choose $\mf m$, since each landmark has a symmetry group of the form $(\mf Q, (\mf I - \mf Q)\mf a_i)$ as given in (\ref{eq:symmetry-group}). The answer is that we should take the \textit{intersection} of these symmetry groups; that is, we note that the combined statistical model must be invariant to transformations of $\mathbb R^3$ which fix \textit{both} the landmarks simultaneously. Obviously, these are the rotations of $\mathbb R^3$ about the axis passing through $\mf a_1$ and $\mf a_2$ (assuming $\mf a_1 \neq \mf a_2$). Another way to arrive at this conclusion is to observe that 
\begin{align}
(\mf I - \mf Q)\mf a_1 &= (\mf I - \mf Q)\mf a_2\\
\Rightarrow \mf Q(\mf a_1 - \mf a_2) &= \mf a_1 - \mf a_2.
\end{align}
Thus, for the two-landmark problem, our symmetry subgroup $H$ is the one-dimensional subgroup of $SE(3)$ representing rotations about the axis passing through $\mf a_1$ and $\mf a_2$: 
\begin{align*}
    H \coloneq \left\lbrace \left(\mf Q, (\mf I - \mf Q)\mf a_1\right)\,\big|\, \mf Q = \exp\big(t\, (\mf a_1 - \mf a_2)^\wedge\big), \,t \in \mathbb R\right\rbrace,
\end{align*}
where $\exp$ refers (by an abuse of notation) to the exponential map of $SO(3)$.
Note that $n_H=1$ and $n_{\Theta}= 5$ in this case. Given three or more landmarks that are not all co-linear, the symmetry group $H$ will be trivial (i.e., $H= \lbrace e \rbrace$, $\mf h = \lbrace \mf 0 \rbrace$).


\section{Localization using Relative Distances}

\subsection{Sensor Network Model}
Let $(\mathcal V, \mathcal E)$ represent a graph with vertices $\mathcal V$ and edges $\mathcal E \subseteq \mathcal V \times \mathcal V$. The vertices are labeled as
$\mathcal V =\lbrace 1, 2, \ldots |\mathcal V|\rbrace$, where $|\,\cdot\,|$ denotes the cardinality of a set, and the edge connecting agents $i$ and $j$ is written as $(i,j)$. The state of the $i^{th}$ agent is represented by the $SE(d)$ transformation matrix $g_i$, which we will also write as $(\mf R_i, \mf p_i)$, where $\mf R_i \in SO(d)$ is the rotation matrix and
$\mf p_i \in \mathbb R^d$ is the translation vector corresponding to $g_i$. In typical engineering applications, we have $d=2$ or $3$.
As explained in the paper, $\mf p_i$ represents the agent's position whereas the role of the rotation matrix $\mf R_i$ is auxiliary; it is needed to describe the underlying homogeneous space as a coset space, but does not necessarily have a physical significance.
We have the matrix representation of $g_i$:
\begin{align}
    g_i = \begin{bmatrix}
    \mf R_i & \mf p_i\\
    \mf 0 & 1
    \end{bmatrix} \in SE(2).
\end{align}
The collective configuration of the sensor network is given by $g=(g_1, g_2, \ldots, g_{|\mathcal V|}) \in SE(2)^{|\mathcal V|}$.

To simplify the presentation, we assume that each edge in $\mathcal E$ generates a single measurement. Thus, the sample space $\mathcal X$ is the space $\mathbb R^{|\mathcal E|}$. A measurement $x\in \mathcal X$ is of the form $x = [\,\cdots \,x_{ij} \cdots\,]^\top$, where $i$ and $j$ range over the values satisfying $(i,j)\in \mathcal E$.

Let $\varphi: SE(d) \rightarrow \mathbb R_{\geq 0}$ be the mapping $(\mathbf R, \mathbf p) \mapsto \frac{1}{2}\lVert \mf p \rVert^2$, with $\lVert \, \cdot \, \rVert$ denoting the vector 2-norm (i.e., the Euclidean norm). 
The distance between agents $i$ and $j$ is given by
\begin{align}
    \varphi (g_i^{-1}g_j) = \frac{1}{2}\lVert \mf p_i - \mf p_j \rVert^2.
\end{align}
Letting $x_{ij}$ represent the distance measured along the edge $(i,j)\in \mathcal E$. If the measurements are unimodal, and the measurement obtained at each edge is independent of those obtained at the other edges, then the measurements may be modeled as follows \cite{chirikjian2014gaussian}:
\begin{align}
    f(x|g) \propto \exp\bigg(-\sum_{(i,j) \in \mathcal E}\frac{1}{2\sigma_{ij}^2} \Big(x_{ij} - \varphi (g_i^{-1} g_j)\Big)^2\bigg),
\end{align}
which assumes (for simplicity of presentation) that the measurement noise at each edge is effectively Gaussian and independent of the measurements at the other edges.
It is readily verified that $\varphi (g_i^{-1}g_j) = \varphi ((hg_i)^{-1}(hg_j))$
for all $i,j \in \mathcal V$.
This observation leads us to consider the following subgroup of $SE(d)^{|\mathcal V|}$, representing the group of symmetries w.r.t. which the measurement model is invariant:
\begin{align}
    \tilde H \coloneqq \lbrace (h,h, \ldots,h) \, |\, h \in SE(d) \rbrace \; \subseteq \; SE(d)^{|\mathcal V|}.
    \label{eq:subgroup}
\end{align}
The subgroup $\tilde H$ corresponds to rigid translations and rotations of the entire sensor network, which leaves the pairwise distances between the agents invariant.
However, there are additional symmetries which must be included, each corresponding to a mapping of the form $(\mf R_i, \mf p_i) \mapsto (\mf Q \mf R_i, \mf p_i)$ that does not change agent $i$'s position in $\mathbb R^d$. Let $H$ denote the subgroup of $SE(d)^{|\mathcal V|}$ which includes these additional symmetries in addition to those in $\tilde H$.
Thus, we have a parameter estimation problem on the homogeneous space $H \backslash SE(d)^{|\mathcal V|}$. 

\subsection{Lie Derivatives of the Log-Likelihood}
The Lie algebra of $SE(2)^{|\mathcal V|}$ is given by $\mathfrak{se}(2)^{|\mathcal V|}$. Let $\mf X_i$ represent an element in $\mathfrak{se}(2)^{|\mathcal V|}$ which is $0$ everywhere but for the component corresponding to agent $i$.
We are interested in derivatives of the form $\mf X_i^R \ell\,(g)$. 
Let the $i^{th}$ block of $\mf X_i$ be denoted as\footnote{In the paper, we set $\epsilon_{\omega}$ to $0$ since our choice of basis for $\mf m$ only has translational components. Nevertheless, the general case is provided here for the sake of completeness.}
\begin{align}
    X_i \coloneq \begin{bmatrix}
    0 & -\epsilon_\omega & \epsilon_{\textrm{x}}\\
    \epsilon_\omega & 0 & \epsilon_{\textrm{y}}\vspace{2pt}\\
    0 & 0 & 0
    \end{bmatrix} \in \mathfrak{se}(2).
\end{align}
Then, we have 
\begin{align}
    \mf X_i^R \ell\,(g) &= X_i^R\left(-\sum_{k\in \mathcal N_i}\frac{1}{2\sigma_{ik}^2} \Big(x_{ik} - \varphi \left((\,\cdot\,)^{-1} g_k\right)\Big)^2\right)(g_i) \nonumber 
\end{align}
where $\mathcal N_i$ is the set of neighbors of agent $i$ in the graph, since these correspond precisely to the terms of $\ell$ that depend on $g_i$. The derivative is evaluated as
\begin{align}
    &\frac{d}{dt} \left(-\sum_{k\in \mathcal N_i}\frac{1}{2\sigma_{ik}^2} \Big(x_{ik} - \varphi \left((\exp(tX_i)g_i)^{-1} g_k\right)\Big)^2\right)\Bigg\vert_{t=0}\nonumber\\
    &\ =-\sum_{k\in \mathcal N_i}\frac{1}{2\sigma_{ik}^2} \,\frac{d}{dt}\Big(x_{ik} - \varphi \left((\exp(tX_i)g_i)^{-1} g_k\right)\Big)^2 \Bigg\vert_{t=0}\\
    &\ =\sum_{k\in \mathcal N_i}\frac{1}{\sigma_{ik}^2} \big(x_{ik} - \varphi (g_i^{-1} g_k)\big)\frac{d}{dt}\varphi \left((\exp(tX_i)g_i)^{-1} g_k\right) \Bigg\vert_{t=0}
    \label{eq:lie-derivative}
\end{align}
Using the fact that $\exp(tX_i) = \mf I + t X_i + O(t^2)$, we have
\begin{align}
    \frac{d}{dt} \varphi &\left(\big(\exp(tX_i)g_i\big)^{-1} g_k\right)\Big\vert_{t=0} \nonumber\\
    &= \frac{d}{dt} \varphi \left(\big((\mf I + t X_i)g_i\big)^{-1} g_k\right)\Big\vert_{t=0} \\
    &= \frac{d}{dt}\frac{1}{2}\left\lVert \mf p_i  + t \epsilon_\omega \begin{bmatrix}
        0 & -1\\
        1 & 0 
    \end{bmatrix}\mf p_i 
    + t \begin{bmatrix}\epsilon_{\textrm x}\\\epsilon_{\textrm y} \end{bmatrix} 
    - \mf p_k \right\rVert^2\Bigg\vert_{t=0}\\
    &= \left(\mf p_i  - \mf p_k \right)^\top\left(\epsilon_\omega \begin{bmatrix}
        0 & -1\\
        1 & 0 
    \end{bmatrix}\mf p_i + 
    \begin{bmatrix}\epsilon_{\textrm x}\\\epsilon_{\textrm y} \end{bmatrix}
        \right).
    \label{eq:varphi-derivative}
\end{align}
Plugging (\ref{eq:varphi-derivative}) into (\ref{eq:lie-derivative}), we get the value of $\mf X_i^R \ell\,(g)$. In practice, we compute $\mf X_i^R \ell\,(\hat g)$ given an estimate $\hat g$ of the true configuration $g$.


\subsection{Choosing a Basis for $\mf m$}
Recall that the dimensions of $SE(d)$ and $SO(d)$ are $\binom{d+1}{2}$ and $\binom{d}{2}$, respectively, where $\binom{n}{k}$ represents the binomial coefficient with indices $n$ and $k$.
Therefore, the dimension of $G$ is $n_G = \binom{d+1}{2}|\mathcal V|$, while the dimension of $H$ is $n_H = \binom{d}{2}|\mathcal V| + \binom{d+1}{2}$ corresponding to the rotations of the individual agents and the subgroup described in (\ref{eq:subgroup}), respectively. The dimension of $\Theta$ (which is also the dimension of $\mf m$) is computed as $n_{\Theta} = n_G - n_H$, which is found to be $n_{\Theta}=d|\mathcal V| - \binom{d+1}{2}$. A question then arises: how do we choose the $n_{\Theta}$ basis vectors of $\mf m$? 

Specializing to the case of $d=2$, we see that $\mf m$ is a $2|\mathcal V| - 3$ dimensional vector space. One way to proceed is to first choose the $2|\mathcal V|$ vectors in $\mathfrak g$ that represent the $2$-dimensional translation of each of the $|\mathcal V|$ agents. Thereafter, we discard the first $3$ vectors, yielding $2|\mathcal V| - 3$ basis vectors as desired. The implication of this choice of basis is that, when taking directional derivatives of $\ell$, the first agent's position is held fixed while only the $\textrm{x}$-coordinate of the second agent is allowed to vary. 

The derivation of the {\bf FIM} proceeds identically to the derivation in \cite{le2018localizability} (while using the formulae given in the previous section), so it is omitted here. Note that the aforementioned work introduced additional constraints (in the form of \textit{anchors}, i.e., agents that can measure their own position) to ensure the invertibility of the {\bf FIM}. 
In contrast, our theory allows us to consider different $n_{\Theta}$-dimensional bases for $\mf m$, each of which yields a different Fisher scoring algorithm, but will nonetheless ensure that the parameter estimation problem is well-posed and the {\bf FIM} invertible.

\subsection{Choosing another Basis for $\mf m$}\label{sec:another-basis}
The choice of basis given in the previous section is adequate for computing the FIM and CRB. However, there is a subtle issue with it that makes it undesirable for the purpose of implementing the Fisher scoring algorithm.
The issue arises when agents $1$ and $2$ share the same $\textrm{y}$-coordinate (i.e., are horizontally aligned). In this special case, perturbing the $\textrm{y}$-coordinate of agent $2$ will not change its distance from the last agent. 
It must be emphasized that this is \textit{not} a limitation of our theory; in fact, our theory helps explain why the above choice of basis fails. It is because the submanifold $\mathcal S_g$ can sometimes degenerate into a lower-dimensional manifold, or it may partially align itself along $\mathfrak h$. More precisely, the issue is that the ``basis" we picked for $\mf m$ either has linear dependence amidst its vectors, or it has linear dependence with some vectors in $\mathfrak h$; therefore, it is not actually a basis for $\mf m$.

One way to fix the aforementioned issue is to assume that the first two agents are vertically aligned, which is what we have done in the main paper. Another possibility is to make the following choice of basis for $\mf m$.
We choose a set of $n_{\Theta}$ edges in $\mathcal E$, and let the basis vectors $\lbrace \mf E_i \rbrace_{i=1}^{n_{\Theta}}$ represent the shrinking/expansion of these edges. That is, the $i^{th}$ basis vector $\mf E_i$ represents the $\mathfrak{se}(2)^{|\mathcal V|}$ element given by
\begin{align}
    \mf E_i = 
\begin{tikzpicture}[baseline=(current bounding box.center)]
    \matrix (m) [matrix of math nodes, nodes in empty cells, left delimiter={[}, right delimiter={]}] {
        \vdots \\
            \mf 0 \\
            \mf V_{p,q} \\
            \mf 0 \\
            \vdots \\
            \mf 0 \\
            \mf V_{q,p} \\
            \mf 0 \\
            \vdots \\
            };
            \node [right=of m-3-1, yshift=-0.5ex] (mth) {$p^{th}$ block};
            \node [right=of m-7-1, yshift=-0.5ex] (nth) {$q^{th}$ block};
            \draw[<-] (mth.west) -- ++(-0.5,0) |- (m-3-1.east);
            \draw[<-] (nth.west) -- ++(-0.5,0) |- (m-7-1.east);
\end{tikzpicture}
\end{align}
for some edge $(p,q)\in\mathcal E$, where
\begin{align}
    \mf V_{p,q} = \begin{bmatrix}
        \mf 0 & \mf p_p - \mf p_q\\
        \mf 0 &  0
    \end{bmatrix} \quad  \textit{and}  \quad  \mf V_{q,p} = -\mf V_{p,q}.
\end{align}

As it turns out, it is possible to choose $2|\mathcal V|-3$ edges in $\mathcal E$ in a way that the resulting basis vectors are linearly independent {if and only if} the sensor network's configuration satisfies a condition called \textit{rigidity}; this is a well-known result in the sensor network localization literature. The corresponding subgraph of $\mathcal G$ which is induced by our choice of $2|\mathcal V|-3$ edges is called a \textit{minimally rigid subgraph}. This choice of basis gives us yet another connection to the existing literature on rigidity theory.

With the choice of basis considered in this section, we can once again compute the gradient of $\ell$ as well as the {\bf FIM}. We have,
\begin{align}
    &\mf E_i^R \ell\,(g)\nonumber\\
    =&
    \sum_{k\in \mathcal N_m}\frac{1}{\sigma_{pk}^2} \big(x_{pk} - \varphi (g_p^{-1} g_k)\big)
    \left(
    \mf p_p - \mf p_k
        \right)^\top\left(
    \mf p_p - \mf p_q
        \right)\nonumber\\
    &+\sum_{k\in \mathcal N_n}\frac{1}{\sigma_{qk}^2} \big(x_{qk} - \varphi (g_q^{-1} g_k)\big)\left(
        \mf p_q - \mf p_k
            \right)^\top
\left(
\mf p_q - \mf p_p
    \right)
\end{align}
Now, let $\mf E_j$ represent another basis vector. If this vector correponds to the edge $(l,m)$, then the derivation of the {\bf FIM} proceeds on a case-by-case basis depending on whether the edge $(l,m)$ is incident on the agents $p$ and $q$ or their neighbors. 

This new {\bf FIM} will be different from the {\bf FIM} presented in the main paper, simply because we have chosen a different basis for $\mf m$. The underlying FIM (where `M' stands for Metric, rather than Matrix) is the same; the FIM is an inherent feature of the statistical model that is independent of how the model is parametrized.

\bibliographystyle{IEEEtran}
\bibliography{../references}